\documentclass[11pt]{article}
\usepackage[margin=1in]{geometry}
\usepackage[final]{microtype}
\usepackage{epsfig}
\usepackage{graphics}
\usepackage{latexsym}
\usepackage{amsmath}
\usepackage{amsfonts}
\usepackage{amssymb}
\usepackage{mathrsfs}
\usepackage{xspace}
\usepackage{epstopdf}
\usepackage{float}
\usepackage{hyperref}
\usepackage{pgfplots}
\usepackage{thmtools}
\usepackage{thm-restate}
\usepackage{longtable}

\numberwithin{equation}{section}
\usepackage[ruled,vlined]{algorithm2e}
\SetArgSty{textrm}

\usepackage[english]{babel}
\usepackage[nottoc]{tocbibind}

\usepackage{caption}
\usepackage{subcaption}

\usepackage[usenames,dvipsnames]{xcolor}
\usepackage{pifont}
\usepackage{booktabs}
\usepackage{multirow}

\newcommand{\bluecomment}[1]{\textcolor{blue}{\textrm{#1}}}

\usepackage{amsthm}     
\theoremstyle{plain}     
\newtheorem{theorem}{Theorem}

\newtheorem{lemma}{Lemma}
\newtheorem{proposition}{Proposition}
\newtheorem{mechanism}{Mechanism}
\theoremstyle{definition} 

\newtheorem{example}{Example}

\theoremstyle{remark} 

\makeatletter
\@addtoreset{equation}{section}
\def\section{\@startsection {section}{1}{\z@}{-3.5ex plus -1ex minus
 -.2ex}{2.3ex plus .2ex}{\large\bf}}
\makeatother

\def\bfm#1{\mbox{\boldmath$#1$}}

\def\0{\bfm 0}

\DeclareMathAlphabet{\mathpzc}{OT1}{pzc}{m}{it}

\newcounter{my}

\newcounter{my2}

\newcounter{my3}

\newcounter{my4}

\newcounter{my5}

\newcounter{my6}

\allowdisplaybreaks 

\begin{document}

\title{A Randomized $\frac32$-Approximation for Strategic Facility Location\\ on a Circle}

\author{Hau Chan$^{1}$\quad Jianan Lin$^{2}$\quad Chenhao Wang $^{3,4}$\\[0.75em]
$1$ University of Nebraska-Lincoln\\
$2$ Rensselaer Polytechnic Institute\\
$3$ Beijing Normal University-Zhuhai\\
$4$ Beijing Normal-Hong Kong Baptist University
}
\date{}
\maketitle
\begingroup
\renewcommand{\thefootnote}{}
\footnotetext{The main body of this work was completed in June 2026, with assistance from GPT-5.4 in overcoming some of the difficulties in the approximation-ratio proofs. The underlying RD--PCD mixture comes from Rogowski and Dziubi{\'n}ski's IJCAI 2025 paper~\cite{rogowski2025improved}. After studying their analysis, we recognized that exploiting the discrepancy introduced by replacing circle distances with line distances, which their analysis leaves unexploited, was the key to a sharper guarantee. The paper was submitted to WINE 2026 and rejected with review scores of 3, 4, and 5. Version~1 was uploaded to arXiv on the day of the rejection; the main results therefore date from June 2026 rather than from this revision. Version~2 revisits the lower-bound argument, with assistance from GPT Astra-6. The original proof relied on a result from the literature that we have not been able to fully verify. We therefore replace it with an independently verifiable, computer-assisted lower bound that does not rely on that result.}
\endgroup


\begin{abstract}
We study strategyproof mechanisms for locating a single facility on a circle so as to serve a set of strategic agents under the utilitarian social cost objective. 
We analyze a simple parity-dependent mechanism that randomizes between the Random Dictator (RD) mechanism and the Proportional Circle Distance (PCD) mechanism. 
For an odd number \(n\) of agents, this mechanism mixes RD and PCD with equal probability, as originally proposed by Rogowski and Dziubi{\'n}ski (IJCAI 2025). 
For even \(n\), it mixes RD with a random-deletion extension of PCD, in which one agent is removed uniformly at random before PCD is applied to the remaining agents.
Our main result shows that this mechanism is strategyproof and achieves an approximation ratio of \(\frac32\) for every \(n\ge 3\). 
This improves the \(\frac74\) upper bound of Rogowski and Dziubi{\'n}ski, which applied only to odd \(n\), and extends the guarantee to even numbers of agents. 
Finally, we establish a computer-assisted lower bound of \(1.088187\ldots\) on the approximation ratio of any randomized strategyproof mechanism, improving on the previous lower bound of \(1.0456\) due to Meir (SAGT 2019). 
\end{abstract}

\section{Introduction}


The problem of designing approximately optimal strategyproof mechanisms for locating a single facility in a metric space has received significant theoretical interest over the past few decades \cite{alon2010strategyproof,gravin2025approximation,moulin1980strategy,procaccia2013approximate,schummer2002strategy}. 
In this problem, a mechanism elicits agents' preferred facility locations and selects a facility location that approximately optimizes a given objective, while ensuring that agents have incentives to report their preferences truthfully. 
Beyond its theoretical interest, the problem is motivated by applications such as choosing locations for public facilities, including schools, parks, and libraries, to serve residents in a region \cite{drezner2004facility}; selecting representatives that reflect the political views of groups \cite{black1948rationale,moulin1980strategy}; and choosing common time points in cyclic domains to coordinate recurring activities \cite{peters2020preferences} while incorporating the preferences of the respective agents.


The study of approximate mechanism design without money for facility location began with the setting of locating a single facility on a line under the utilitarian social cost objective, defined as the sum of the distances between the agents' preferred locations and the facility location \cite{procaccia2013approximate}. 
In this setting, Procaccia and Tennenholtz \cite{procaccia2013approximate} showed that the median mechanism, which locates the facility at the median reported location, is strategyproof and optimal.
Building on this work, subsequent studies \cite{alon2010strategyproof,dokow2012mechanism} considered single-facility location under the utilitarian social cost objective on more general graphs, including trees and circles.
For trees, Alon et al.~\cite{alon2010strategyproof} showed that the median mechanism admits a natural tree-median generalization that remains strategyproof and exactly minimizes the social cost. 
For circles, however, every deterministic strategyproof and onto mechanism is dictatorial \cite{schummer2002strategy}, and can have approximation ratio \(n-1\) with \(n\) agents \cite{alon2010strategyproof}.

For facility location on a circle, subsequent studies \cite{alon2010strategyproof,DBLP:conf/sagt/Meir19,rogowski2025improved} have examined randomized strategyproof mechanisms in order to improve over deterministic mechanisms. 
Alon et al.~\cite{alon2010strategyproof} showed that the simple \emph{Random Dictator} (RD) mechanism, which returns each reported agent location with probability \(\frac1n\), achieves an approximation ratio of \(2-\frac{2}{n}\) for every \(n\ge 2\). 
Meir~\cite{DBLP:conf/sagt/Meir19} introduced two randomized mechanisms that beat RD on the circle when \(n=3\): the \emph{Proportional Circle Distance} (PCD) mechanism, which selects each reported location \(x_i\) with probability proportional to the length \(L_i\) of the arc facing agent \(i\); and the \emph{\(q\)-Quadratic Circle Distance} (\(q\)-QCD) mechanism, in which the probability of selecting \(x_i\) is proportional to \(\max\{L_i^2,q^2\}\). 
For $n=3$, Meir proved that PCD and \(\frac14\)-QCD are strategyproof and achieve approximation ratios \(\frac54\) and \(\frac76\), respectively. 
Rogowski and Dziubi{\'n}ski \cite{rogowski2025improved} proposed a mechanism that mixes RD and PCD with equal probability and derived an improved approximation ratio of \(\frac74\) for odd numbers \(n\ge 3\) of agents. They also provided computational evidence conjecturing that the exact approximation ratio of this mixture is \(\frac32\).
Farjoun and Meir~\cite{farjoun2025strategyproof} subsequently proved that PCD has a tight approximation ratio of $7-4\sqrt2\approx1.34315$ for $n=5$. Their analysis gives a sharper guarantee for this fixed population size, complementary to our uniform guarantee for arbitrary odd and even $n$.
On the inapproximability side, Meir~\cite{DBLP:conf/sagt/Meir19} proved a lower bound of \(1.0456\) on the approximation ratio of any randomized strategyproof mechanism.

\subsection{Our Results}\label{subsec:contributions}

Motivated by the remaining gap between the best known upper and lower bounds, and by the absence of a unified guarantee for arbitrary numbers of agents, we revisit the design and analysis of randomized strategyproof mechanisms on the circle. 
We show that a simple mixed mechanism achieves an approximation ratio of \(\frac32\) for every \(n\ge 3\). 
This improves the previous \(\frac74\) upper bound, which applied only to odd \(n\), and extends the guarantee to even numbers of agents.
We also establish an improved computer-assisted lower bound of \(1.088187\ldots\) on the approximation ratio of any randomized strategyproof mechanism on the circle. 
Thus, our results narrow the gap between the best known upper and lower bounds for randomized strategic facility location on the circle from both directions. Figure~\ref{fig:upper-bound-comparison} summarizes the bounds.

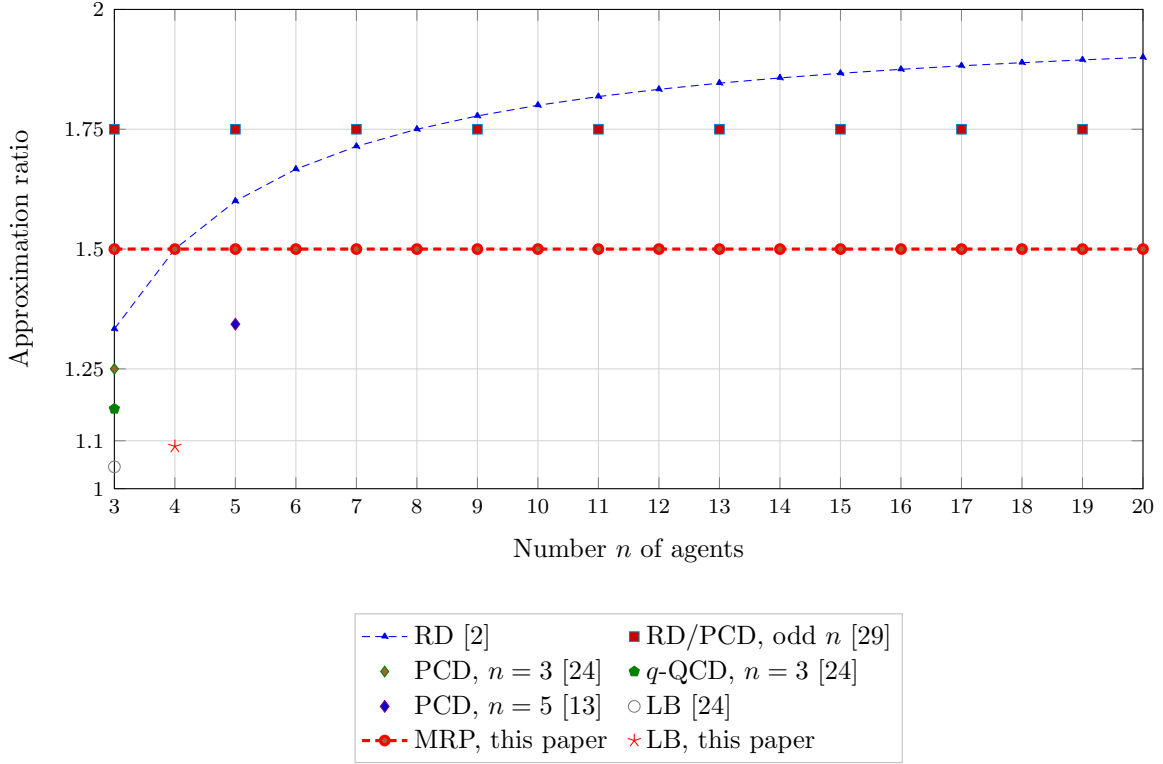
\begin{figure}[!htb]
    \centering
    \begin{tikzpicture}
    \begin{axis}[
        width=0.92\linewidth,
        height=0.48\linewidth,
        xmin=3, xmax=20,
        ymin=1.0, ymax=2.0,
        xlabel={Number \(n\) of agents},
        ylabel={Approximation ratio},
        xtick={3,4,...,20},
        ytick={1.0,1.1,1.25,1.5,1.75,2.0},
        tick label style={font=\scriptsize},
        label style={font=\small},
        legend style={
            font=\small,
            at={(0.5,-0.26)},
            anchor=north,
            legend columns=2,
            cells={anchor=west},
            fill=white,
            draw=gray!50,
            inner sep=2pt,
            /tikz/every even column/.append style={column sep=0.5em}
        },
        legend image post style={scale=1},
        grid=both,
        grid style={line width=.1pt, draw=gray!25},
        major grid style={line width=.2pt, draw=gray!35},
    ]

    \addplot+[blue, densely dashed, mark=triangle*, mark size=1.4pt]
    coordinates {
        (3,1.3333) (4,1.5000) (5,1.6000)
        (6,1.6667) (7,1.7143) (8,1.7500) (9,1.7778)
        (10,1.8000) (11,1.8182) (12,1.8333) (13,1.8462)
        (14,1.8571) (15,1.8667) (16,1.8750) (17,1.8824)
        (18,1.8889) (19,1.8947) (20,1.9000)
    };
    \addlegendentry{RD~\cite{alon2010strategyproof}}

    \addplot+[cyan!65!blue, only marks, mark=square*, mark size=1.8pt]
    coordinates {
     (3,1.75)   (5,1.75) (7,1.75) (9,1.75) (11,1.75)
        (13,1.75) (15,1.75) (17,1.75) (19,1.75)
    };
    \addlegendentry{RD/PCD, odd \(n\)~\cite{rogowski2025improved}}

    \addplot+[green!50!black, only marks, mark=diamond*, mark size=2.0pt]
    coordinates {(3,1.25)};
    \addlegendentry{PCD, \(n=3\)~\cite{DBLP:conf/sagt/Meir19}}

    \addplot+[green!50!black, only marks, mark=pentagon*, mark size=2.0pt]
    coordinates {(3,1.1667)};
    \addlegendentry{\(q\)-QCD, \(n=3\)~\cite{DBLP:conf/sagt/Meir19}}

    \addplot+[violet, only marks, mark=diamond*, mark size=2.2pt]
    coordinates {(5,{7-4*sqrt(2)})};
    \addlegendentry{PCD, \(n=5\)~\cite{farjoun2025strategyproof}}

    \addplot+[gray, only marks, mark=o, mark size=2.2pt]
    coordinates {(3,1.0456)};
    \addlegendentry{LB~\cite{DBLP:conf/sagt/Meir19}}

    \addplot+[red, very thick, mark=*, mark size=1.6pt]
    coordinates {
        (3,1.5) (4,1.5) (5,1.5) (6,1.5) (7,1.5)
        (8,1.5) (9,1.5) (10,1.5) (11,1.5) (12,1.5)
        (13,1.5) (14,1.5) (15,1.5) (16,1.5) (17,1.5)
        (18,1.5) (19,1.5) (20,1.5)
    };
    \addlegendentry{MRP, this paper}

    \addplot+[red, only marks, mark=star, mark size=2.6pt]
    coordinates {(4,1.088187342345)};
    \addlegendentry{LB, this paper}

    \end{axis}
    \end{tikzpicture}
    \caption{Comparison of upper-bound guarantees and lower-bound witnesses for randomized strategic facility location on the circle. The red horizontal line is the uniform $\frac32$ upper bound for MRP, not its exact approximation ratio for each fixed $n$. The lower-bound markers indicate witness population sizes, rather than bounds established separately for every $n$.}
    \label{fig:upper-bound-comparison}
\end{figure}

Our mechanism depends on the parity of the number $n$ of agents; we call it the \emph{Mixed RD\&PCD} (MRP) mechanism. For odd \(n\), the MRP mechanism randomly selects 
RD and PCD with equal probability; this is the same mechanism studied by Rogowski and Dziubi{\'n}ski~\cite{rogowski2025improved}. {For even \(n\), we first introduce the \emph{Even Proportional Circle Distance} (EPCD) mechanism, a random-deletion extension of PCD: it first deletes one agent uniformly at random and then applies PCD to the remaining odd-agent profile. The MRP mechanism then mixes 
RD and EPCD with equal probability. }


Our main contribution is to demonstrate that the MRP mechanism achieves an approximation ratio of \(\frac32\) for every \(n\ge 3\), while the cases \(n=1,2\) are trivial. This improves the previous \(\frac74\) guarantee for odd \(n\) due to Rogowski and Dziubi{\'n}ski~\cite{rogowski2025improved}, thereby resolving their \(\frac32\)-approximation conjecture for this mechanism. It also extends the same guarantee to even numbers of agents, therefore yielding a unified guarantee based on a common framework for all nontrivial values of \(n\).

Establishing this guarantee is technically challenging for two reasons. 
First, unlike the line, the circle has no canonical left-to-right order or boundary point. 
Any analysis that cuts the circle into a line necessarily introduces an artificial discontinuity, and distances between agents on opposite sides of the cut may be overestimated. 
Thus, a line-based argument cannot be applied directly without carefully accounting for the error caused by the cut. 
Second, the PCD mechanism is highly sensitive to the circular ordering of the agents: its selection probabilities depend on the lengths of opposite arcs, which change nonlinearly with the agent configuration. 
This makes it difficult to compare the expected cost of PCD with the optimum by a purely local or pairwise argument. 
The even case introduces an additional difficulty, since PCD is naturally defined for odd profiles; deleting a random agent changes both the circular order and the arc lengths, so the analysis must control the expected effect of this deletion.

The main technical contribution is to sharpen the cut-based analysis of Rogowski and Dziubi{\'n}ski~\cite{rogowski2025improved}. Assume the circle is of length 1. 
Their analysis first uses anonymity and neutrality to restrict attention to normalized profiles in which an optimal location is fixed, and then cuts the unit circle at a carefully chosen point. 
This replaces the circular metric by the line-segment distance \(d'(u,v)=|u-v|\), thereby avoiding the nonlinearity caused by the minimum in the circular distance \(d(u,v)=\min\{|u-v|,1-|u-v|\}\). 
The resulting line metric is easier to analyze, but it is also conservative: it can strictly overestimate the true circular distance, precisely for pairs of points whose shortest path crosses the cut. 
In their  \(\frac74\)-approximation analysis, this loss is not exploited.
Our analysis keeps track of the discrepancy between the cut metric and the original circular metric. 
The key observation is that the overestimation introduced by the cut can be charged against the distance savings provided by the PCD component of the mechanism. 
This requires a more delicate accounting of which pairs of agents are separated by the cut and how the PCD probabilities weight the corresponding arcs. 
By refining this distance calculation, we improve the odd-agent analysis from \(\frac74\) to \(\frac32\). 
We then extend the argument to even numbers of agents through the random-deletion EPCD mechanism, controlling in expectation how the deletion step changes the relevant arcs and crossing terms. 
Consequently, the MRP mechanism achieves an approximation ratio of \(\frac32\) for every \(n\ge3\).


 Finally, we complement our upper-bound result with two lower-bound results. 
First, we show that the \(\frac32\) guarantee is asymptotically tight for the MRP mechanism, for both odd and even numbers of agents. 
These tight examples also explain why the mechanism mixes RD with PCD, or with EPCD in the even case, with equal probability. 
Second, we establish a computer-assisted lower bound of \(1.088187\ldots\) for general randomized strategyproof mechanisms by constructing four-agent instances. 
This improves the previously best known lower bound of \(1.0456\) due to Meir~\cite{DBLP:conf/sagt/Meir19}. 
Therefore, the current gap for randomized strategyproof facility location on the circle is between \(1.088187\ldots\) and \(\frac32\).


\subsection{Other Related Work}

When designing strategyproof mechanisms for locating a single facility in a metric space, existing related  studies have mostly focused on characterizing strategyproof mechanisms independent of the cost objective or identifying strategyproof mechanisms that approximately optimize a given objective. 
Below, we review studies that examine these two directions in metric spaces closely related to the  circle setting considered here. 
We refer readers to the survey of \cite{chan2021mechanismsurvey} for other mechanism design settings and variants of facility location. 

\paragraph{Characterizations.}

Characterizations of strategyproof mechanisms for facility location have been studied in several metric spaces. 
Moulin~\cite{moulin1980strategy} characterized strategyproof mechanisms for single-peaked preferences on the line by generalized median rules with phantom voters, and Border and Jordan~\cite{border1983} studied related phantom-voter characterizations in multidimensional settings. 
For network domains, Schummer and Vohra~\cite{schummer2002strategy} characterized strategyproof mechanisms and showed that circles impose strong restrictions on deterministic mechanisms. 
Dokow et al.~\cite{dokow2012mechanism} further studied mechanism design on discrete lines and  sufficiently large cycles. 
Miyagawa~\cite{DBLP:journals/scw/Miyagawa01} and Fotakis and Tzamos~\cite{fotakis2013winner} studied characterizations for multiple-facility mechanisms on the line, while Aziz et al.~\cite{aziz2020capacity} characterized strategyproof mechanisms with capacity constraints. 
More recently, Tang et al.~\cite{tang2020characterization} characterized group-strategyproof mechanisms in strictly convex spaces, and Lin~\cite{lin2020nearly} gave a characterization for two-agent 
deterministic strategyproof mechanisms in \(L_p\) spaces. 

\paragraph{Utilitarian social cost minimization.}

Procaccia and Tennenholtz~\cite{procaccia2013approximate} introduced the framework of approximate mechanism design without money, using facility location as a central example. On the line, the median mechanism is strategyproof and minimizes the sum of agents' distances to the facility. 
Lu et al.~\cite{lu2009tighter,lu10mechanism} studied deterministic and randomized mechanisms for locating two facilities on the line.
The utilitarian objective has also been studied in multidimensional spaces. Meir~\cite{DBLP:conf/sagt/Meir19} showed that the coordinate-wise median (CM) mechanism achieves a $\sqrt d$-approximation in $d$-dimensional Euclidean space. Goel and Hann-Caruthers~\cite{GoelH23} established its optimality among deterministic anonymous strategyproof mechanisms in the Euclidean plane, and Gravin and Jia~\cite{gravin2025approximation} obtained dimension-independent guarantees for CM in $L_q$ spaces. More recently, Barak~\cite{barak2026facility} and, independently, Chan et al.~\cite{chan2026strategyproof} established a randomized $\frac4\pi$-approximation in the Euclidean plane by applying CM after a uniformly random rotation. Barak also analyzed this mechanism in higher-dimensional Euclidean spaces. Beyond the utilitarian objective, Chan et al.~\cite{chan2026strategyproof} and Hastings~\cite{hastings2026strategic} studied approximation guarantees under more general $L_p$-norm social costs.

For arbitrary numbers of facilities on the line, Fotakis and Tzamos~\cite{fotakis2013strategyproof} introduced the randomized Equal Cost mechanism, which is group strategyproof for concave connection costs and achieves an $n$-approximation for social cost. This includes the standard distance cost as a special case. Several recent concurrent works have improved randomized guarantees for two or three facilities and for other specific choices of $k$. Ma and Peng~\cite{ma2026breaking}, Jia~\cite{jia2026product}, and Chan et al.~\cite{chan2026randomized} obtained approximation ratios below $4$ for two facilities on the line, with the results of Ma and Peng and Chan et al. extending to Ptolemaic metric spaces, including Euclidean spaces. For three facilities on the line, Aziz et al.~\cite{aziz2026anchoring} obtained an $8$-approximation, while Jia~\cite{jia2026product} and Chan et al.~\cite{chan2026randomized} obtained a $6$-approximation. Chan et al. also improved the guarantee for $k=n-1$ facilities in arbitrary metric spaces. Learning-augmented variants have been studied for both single-facility and two-facility settings~\cite{agrawal2022learning,Xu2022}.

\section{Preliminaries}\label{sec:model}

We study the facility location problem on a circle of length \(1\). Let
\(N=\{1,2,\ldots,n\}\) be the set of agents. We denote the circle by \(G\), and
identify it with the interval \([-\frac12,\frac12]\) after gluing the two
endpoints. The distance between two points \(x,y\in G\) is the length of the
shorter arc between them:
\[
    d(x,y)=\min\{|x-y|,1-|x-y|\}.
\]
A location profile is denoted by \(\mathbf{x}=(x_1,\ldots,x_n)\in G^n\), where
\(x_i\in G\) is the preferred location of agent \(i\). For any integer $k\ge 1$, write $[k]=\{1,\ldots,k\}$.

A randomized mechanism is a function
\(f:G^n\to\Delta(G)\), where \(\Delta(G)\) denotes the set of probability
distributions over \(G\). For a randomized mechanism $f$, we write
\(Y\sim f(\mathbf{x})\) for the random facility location sampled from the
distribution returned on profile \(\mathbf{x}\). Given a facility location
\(y\in G\), the cost of agent \(i\) is \(d(x_i,y)\). For a random facility
location \(Y\), the cost of agent \(i\) is the expectation
\(\mathbb E_Y[d(x_i,Y)]\). Strategyproofness requires a mechanism to ensure that no agent can gain by misreporting. Formally, a mechanism \(f\) is
\emph{strategyproof} if, for every profile \(\mathbf{x}\in G^n\), every agent
\(i\in N\), and every alternative report \(x_i'\in G\),
\[
    \mathbb E_{Y\sim f(\mathbf{x})}\big[d(x_i,Y)\big]
    \le
    \mathbb E_{Y'\sim f(x_i',\mathbf{x}_{-i})}\big[d(x_i,Y')\big],
\]
where \((x_i',\mathbf{x}_{-i})\) is the
profile obtained from \(\mathbf{x}\) by replacing only agent \(i\)'s report with
\(x_i'\).

We use the following standard symmetry properties. A mechanism $f$ is said to be
\begin{itemize}
    \item \emph{Anonymous}, if its output distribution is invariant under permutations of the agents. 
\item \emph{Neutral}, if  relabeling the graph's locations in a symmetry-preserving way simply relabels the outcome in the same way. Formally, for any automorphism \(\gamma\) of $G$,  $f(\gamma \circ \mathbf x) =\gamma\circ f(\mathbf x)$. 
\item \emph{Peaks-only}, if
the support of \(f(\mathbf{x})\) is contained in the set of reported locations
\(\{x_1,\ldots,x_n\}\).
\end{itemize}

We evaluate a facility location by utilitarian social cost. For a point
\(y\in G\), let \(\mathrm{SC}(\mathbf{x},y)=\sum_{i\in N}d(x_i,y)\) be the total distance from all agent locations to $y$. For a
 mechanism \(f\), the social cost  achieved on profile \(\mathbf{x}\)
is the expectation
\[
    \mathrm{SC}(\mathbf{x},f(\mathbf{x}))
    =
    \mathbb E_{Y\sim f(\mathbf{x})}
    \left[
        \sum_{i\in N}d(x_i,Y)
    \right].
\]
The optimal social cost on profile \(\mathbf{x}\) is
\(
    \mathrm{OPT}(\mathbf{x})
    =
    \min_{y\in G}\mathrm{SC}(\mathbf{x},y),
\)
where the minimum is always attained since \(G\) is compact and the social cost function is
continuous. The \emph{approximation ratio} of a mechanism \(f\) is
\[
    \rho(f)
    =
    \sup_{\mathbf{x}}
    \frac{\mathrm{SC}(\mathbf{x},f(\mathbf{x}))}{\mathrm{OPT}(\mathbf{x})},
\]
where the profile $\mathbf x$ is taken over all $n$-agent instances with $\mathrm{OPT}(\mathbf{x})>0$, for all $n$.
For profiles with \(\mathrm{OPT}(\mathbf{x})=0\), all agents report the same
point. All mechanisms considered in this paper choose that point with
probability one, so these profiles do not affect the approximation ratio. 

\paragraph{Randomized mechanisms.}
We next present the mechanisms used in the paper. The \emph{Random Dictator}
mechanism chooses one agent uniformly at random and locates
the facility at the report of that agent.  It
was first introduced  by
Alon et al.~\cite{alon2010strategyproof} for facility location on networks, and is known to be strategyproof, anonymous, neutral, and peaks-only.

\begin{mechanism}[Random Dictator ($\mathrm{RD}$)]
  Given location profile $\mathbf x$, select each agent $i\in N$ with probability $\frac1n$ and return the reported location \(x_i\).
\end{mechanism}

The \emph{Proportional Circle Distance} mechanism was introduced by Meir
\cite{DBLP:conf/sagt/Meir19} for an odd number of agents, and it is known to be
strategyproof for any odd $n$, anonymous, neutral, and peaks-only.

\begin{mechanism}[Proportional Circle Distance (\(\mathrm{PCD}\))]
    Let \(n=2k+1\). Fix a cyclic ordering
\(z_1,\ldots,z_{2k+1}\) of the reported copies in clockwise order, with
ties broken consistently. Let \(g_t\) be the clockwise gap from \(z_t\) to
\(z_{t+1}\), where indices are taken modulo \(2k+1\). Thus
\(\sum_{t=1}^{2k+1}g_t=1\). The \(\mathrm{PCD}\) mechanism assigns to each ordered copy
\(z_t\) probability equal to the length of its opposing gap:
\[
    \Pr_{\mathrm{PCD}(\mathbf{x})}[z_t]
    =
    g_{t+k},
    \qquad
    t\in[2k+1].
\]
If several ordered copies correspond to the same location, their
probabilities  are aggregated. 
\end{mechanism}

Figure~\ref{fig:pcd-example} illustrates the opposing gaps and selection probabilities for three agents.

\begin{figure}[htbp]
    \centering
    \begin{minipage}[c]{0.48\linewidth}
    \centering
    \begin{tikzpicture}[scale=1.25]
        \draw[gray!50] (0,0) circle (1.3);
        \draw[blue, very thick, ->] (90:1.3) arc (90:0:1.3);
        \draw[red, very thick, ->] (0:1.3) arc (0:-120:1.3);
        \draw[green!50!black, very thick, ->] (-120:1.3) arc (-120:-270:1.3);
        \node[blue] at (45:1.63) {$g_1=\frac14$};
        \node[red] at (-60:1.68) {$g_2=\frac13$};
        \node[green!50!black] at (165:1.73) {$g_3=\frac5{12}$};
        \fill (90:1.3) circle (1.5pt);
        \fill (0:1.3) circle (1.5pt);
        \fill (-120:1.3) circle (1.5pt);
        \node[above] at (90:1.3) {$z_1=0$};
        \node[right] at (0:1.3) {$z_2=\frac14$};
        \node[below left] at (-120:1.3) {$z_3=-\frac5{12}$};
    \end{tikzpicture}
    \end{minipage}\hfill
    \begin{minipage}[c]{0.48\linewidth}
    \centering
    \begin{tabular}{ccc}
        \toprule
        Report & Opposing gap & Probability \\
        \midrule
        $z_1$ & $g_2$ & $\frac13$ \\[3pt]
        $z_2$ & $g_3$ & $\frac5{12}$ \\[3pt]
        $z_3$ & $g_1$ & $\frac14$ \\
        \bottomrule
    \end{tabular}
    \end{minipage}
    \caption{PCD on a circle of circumference $1$ with three agents. Arrows follow the clockwise order, and each gap is labeled by its arc length. Each report is selected with probability equal to the gap between the other two reports that does not contain it.}
    \label{fig:pcd-example}
\end{figure}

For an even number of agents, we define the following random-deletion extension of
\(\mathrm{PCD}\).

\begin{mechanism}[Even Proportional Circle Distance (\(\mathrm{EPCD}\))]\label{mec:epcd}
    Let \(n=2k\ge 4\). First choose one agent uniformly at random and
delete that agent. Then run \(\mathrm{PCD}\) on the remaining \(2k-1\) agents.
\end{mechanism}
\(\mathrm{EPCD}\) has an equivalent formulation: if \(z_1,\ldots,z_{2k}\) are the ordered reported
copies in clockwise order and \(g_t\) is the clockwise gap from \(z_t\) to
\(z_{t+1}\), then
\begin{equation}\label{eq:epqq}
    \Pr_{\mathrm{EPCD}(\mathbf{x})}[z_t]
    =
    \frac12\left(g_{t+k-1}+g_{t+k}\right),
    \qquad
    t\in[2k],
\end{equation}
where indices are taken modulo \(2k\). This is because for each surviving copy $z_t$, the relevant opposing gap is $g_{t+k-1}$
 for half of the possible deletions and $g_{t+k}$ for the other half, yielding the average. As above, probabilities assigned to
ordered copies at the same location are aggregated. 
Since, after conditioning on the deleted agent, the remaining mechanism is
\(\mathrm{PCD}\), the mechanism \(\mathrm{EPCD}\) is strategyproof for any even $n$, anonymous, neutral,
and peaks-only.

For any two randomized mechanisms \(f\) and \(h\), and for \(\lambda\in[0,1]\), we
write \(\lambda f+(1-\lambda)h\) for the mechanism that runs \(f\) with
probability \(\lambda\) and \(h\) with probability \(1-\lambda\). The odd-agent
mixture \(\frac12 \mathrm{RD}+\frac12 \mathrm{PCD}\) was studied by Rogowski and Dziubi{\'n}ski
\cite{rogowski2025improved}. In this paper, we extend it to the even-agent case and analyze the following unified
mechanism.

\begin{mechanism}[Mixed RD\&PCD (MRP)]\label{mec:m}
    Let $n\ge 3$ be the number of agents. Run \[
    \begin{cases}
        \frac12 \mathrm{RD}+\frac12 \mathrm{PCD}, & \text{if } n \text{ is odd} ,\\[1mm]
        \frac12 \mathrm{RD}+\frac12 \mathrm{EPCD}, & \text{if } n \text{ is even}.
    \end{cases}
\]
\end{mechanism}

Since
strategyproofness, anonymity, neutrality, and the peaks-only property are
preserved under fixed-probability mixtures, \(\mathrm{MRP}\) satisfies all these
properties. The case \(n\le 2\) is trivial because any peaks-only mechanism is optimal.

\section{Upper Bound}\label{sec:upper}

In this section we prove the upper bound guaranteed by the unified mechanism
defined in Section~\ref{sec:model}. The argument refines the cutting approach
of Rogowski and Dziubi{\'n}ski~\cite{rogowski2025improved}. Their analysis cuts
the circle at a carefully chosen point and upper-bounds the cost of the
mechanism under the induced line metric. We use the same reduction as a
starting point, but we keep track of the exact loss caused by the cut. The key
point is that the additional cost introduced by the cut is paid for by the
crossing-saving terms of the PCD-type mechanism.
The main result is the following theorem and the remainder of this section is devoted to the proof. 

\begin{theorem}\label{thm:main-upper}
    For every \(n\ge 3\), the mechanism \(\mathrm{MRP}\) is strategyproof,
    anonymous, neutral, and peaks-only, and it achieves an approximation ratio at most $\frac32$ on the circle.
\end{theorem}

The strategyproofness and symmetry properties follow directly from the
corresponding properties of \(\mathrm{RD}\), \(\mathrm{PCD}\), and \(\mathrm{EPCD}\), and from closure under fixed-probability mixtures. 
Before proving the approximation ratio $\frac32$, we first give examples to show the tightness of the approximation analysis and justify the choice of the probabilities in \(\mathrm{MRP}\).


\begin{proposition}[Tightness of \(\mathrm{MRP}\)]\label{prop:tight-examples}
    For arbitrarily large odd and even $n$, there exist instances such that the  ratio between the social cost of \(\mathrm{MRP}\) and the optimal social cost is
      arbitrarily close to \(\frac32\). 
\end{proposition}

\begin{proof}
   Consider the profile with one agent at
    \(-\frac14\), one agent at \(\frac14\), and \(n-2\) agents at \(0\). The optimal
    location is \(0\), and the optimal social cost is \(\mathrm{OPT}=\frac12\). If
    the facility is placed at either nonzero report, the social cost is
    \(\frac12+\frac{(n-2)}{4}=\frac{n}{4}\). Therefore the cost of \(\mathrm{RD}\) is
    \[
        C_{\mathrm{RD}}
        =
        \frac{2}{n}\cdot \frac{n}{4}
        +
        \frac{n-2}{n}\cdot \frac12
        =
        1-\frac1n .
    \]
   For odd \(n=2k+1\) with \(k\ge 2\), the \(\mathrm{PCD}\) mechanism selects a copy of
    \(0\) with probability one, because in the cyclic order  \(-\frac14,0,\ldots,0,\frac14\) all reports with  nonzero opposing gaps are copies of \(0\). 
    Hence \(C_{\mathrm{PCD}}=\mathrm{OPT}\). 
    For even \(n=2k\) with
    \(k\ge 3\), the same conclusion  \(C_{\mathrm{EPCD}}=\mathrm{OPT}\)  follows from the formula
    \(\Pr[z_t]=\frac12(g_{t+k-1}+g_{t+k})\) for \(\mathrm{EPCD}\). 
    Consequently, for both parities, the ratio 
    \[
        \frac{\mathrm{SC}(\mathrm{MRP})}{\mathrm{OPT}}
        =
        \frac{\frac12(C_{\mathrm{RD}}+\mathrm{OPT})}{\mathrm{OPT}}
        =
        \frac32-\frac1n 
    \]
   converges to $\frac32$ as $n\to \infty$. 
\end{proof}

In this example, \(\mathrm{RD}\) is
    a 2-approximation, and \(\mathrm{PCD}\) or \(\mathrm{EPCD}\) is optimal. We further provide  a complementary example in Example~\ref{exam:com}, where \(\mathrm{RD}\) is close to optimal, and
    \(\mathrm{PCD}\) or \(\mathrm{EPCD}\) is a \(2\)-approximation. Together these examples
    explain why the equal probability mixture in our \(\mathrm{MRP}\) mechanism is the natural balance.

\begin{example}\label{exam:com}
     Let  \(n=2k+1\) be odd, and set \(x=1/(4\sqrt{k})\). Consider the profile $\mathbf x$ with
    \(k\) agents at \(-x\), \(k\) agents at \(0\), and one agent at
    \(1/2-x\). The support is contained in a semicircle of length \(1/2\), and
    \(0\) is a median point in this semicircle representation. Hence \(0\) is
    optimal, with
    \[
        \mathrm{OPT}
        =
        kx+\left(\frac12-x\right)
        =
        (k-1)x+\frac12 .
    \]
    Denote by \(L\), \(O\), and \(R\) the locations \(-x\), \(0\), and
    \(1/2-x\), respectively. Their social costs are
    \[
        SC(L,\mathbf x)=\mathrm{OPT}+x,\qquad SC(O,\mathbf x)=\mathrm{OPT},\qquad SC(R,\mathbf x)=k(1-x).
    \]
    From the three nonzero gaps \(x\), \(1/2-x\), and \(1/2\), the aggregated
    \(\mathrm{PCD}\) probabilities of \(L\), \(O\), and \(R\) are respectively
    \(1/2-x\), \(1/2\), and \(x\). Therefore
    \[
        C_{\mathrm{PCD}}
        =
        \left(\frac12-x\right)SC(L,\mathbf x)
        +
        \frac12\mathrm{OPT}
        +
        xSC(R,\mathbf x) = (1-x)\mathrm{OPT} + x\left(\frac12-x\right)+kx(1-x).
    \]
    Since \(\mathrm{OPT}=(k-1)x+1/2\) and \(x=1/(4\sqrt{k})\), we have \(\mathrm{OPT}\sim kx\). Also
    \(C_{\mathrm{PCD}}=\mathrm{OPT}+kx+O(1)\). Hence \(C_{\mathrm{PCD}}/\mathrm{OPT}\to 2\) when $k$ approaches infinity. On the other hand,
    \[
        C_{\mathrm{RD}}
        =
        \frac{kSC(L,\mathbf x)+k\mathrm{OPT}+SC(R,\mathbf x)}{2k+1}
        =
        \frac{2k\mathrm{OPT}+k}{2k+1},
    \]
    and thus \(C_{\mathrm{RD}}/\mathrm{OPT}\to 1\). Therefore the ratio of
    \(\frac12 \mathrm{RD}+\frac12 \mathrm{PCD}\) on this profile converges to \(3/2\).

 The even case has an analogous construction. Let \(n=2k\), set again
    \(x=1/(4\sqrt{k})\), and consider the profile with \(k-1\) agents at
    \(-x\), \(k-1\) agents at \(0\), and two agents at \(1/2-x\). The point
    \(0\) is optimal, and
    \[
        \mathrm{OPT}
        =
        (k-1)x+2\left(\frac12-x\right)
        =
        (k-3)x+1 .
    \]
    With the same notation \(L=-x\), \(O=0\), and \(R=1/2-x\), the social costs
    are
    \[
        SC(L,\mathbf x)=\mathrm{OPT}+2x,\qquad SC(O,\mathbf x)=\mathrm{OPT},\qquad SC(R,\mathbf x)=(k-1)(1-x).
    \]
    A direct application of the \(\mathrm{EPCD}\) probability formula shows that the
    aggregated probabilities of \(L\), \(O\), and \(R\) are again
    \(1/2-x\), \(1/2\), and \(x\), respectively. Hence
    \[
        C_{\mathrm{EPCD}}
        =
        \left(\frac12-x\right)SC(L,\mathbf x)
        +
        \frac12\mathrm{OPT}
        +
        xSC(R,\mathbf x) .
    \]
    Since \(\mathrm{OPT}\sim kx\), we again have \(C_{\mathrm{EPCD}}=\mathrm{OPT}+kx+O(1)\), and so
    \(C_{\mathrm{EPCD}}/\mathrm{OPT}\to 2\). Meanwhile,
    \[
        C_{\mathrm{RD}}
        =
        \frac{(k-1)SC(L,\mathbf x)+(k-1)\mathrm{OPT}+2SC(R,\mathbf x)}{2k}
        =
        \frac{k-1}{k}(\mathrm{OPT}+1),
    \]
    and therefore \(C_{\mathrm{RD}}/\mathrm{OPT}\to 1\). Therefore the ratio of
    \(\frac12 \mathrm{RD}+\frac12 \mathrm{EPCD}\) on this profile converges to \(3/2\).
\end{example}

\subsection{Setup}\label{subsec:main}




We will prove the approximation ratio for both parities separately,
and we begin with two lemmas that will be used in both cases. The first
lemma is a normalization of the profile, such that it suffices to consider those profiles with an optimal location at $0$. Since all mechanisms considered here are
anonymous and neutral, we may rotate the circle so that an optimal location is
the origin, and then relabel agents according to the clockwise order around the
circle. The second lemma allows us to cut the circle at the antipodal point of 0 (i.e., the glued point of $-\frac12$ and $\frac12$) and quantifies the increase of social cost after cutting, compared to the original social cost in the circle.
Frequently used notations in this section are summarized in Table \ref{tab:notation}.

\begin{lemma}[Normalization]\label{lem:normalization}
    Let \(f\) be an anonymous and neutral mechanism. For every profile
    \(\mathbf{x}\) with \(\mathrm{OPT}(\mathbf{x})>0\), there is a profile
    \(\mathbf{b}\) so that
    \(0\) is an optimal location and $f$ achieves the same ratio for $\mathbf b$ and $\mathbf x$. Moreover,  using the
    representatives in \([-\frac12,\frac12]\), the following forms can be
    imposed.
\begin{itemize}
    \item[(i)]  If \(n=2k+1\), then the agents can be indexed as
    \(-k,\ldots,-1,0,1,\ldots,k\) so that
    \[
        b_{-i}=-a_i,\qquad b_0=0,\qquad b_i=c_i,\qquad i\in[k],
    \]
    where \(0\le a_1\le\cdots\le a_k\le \frac12\) and
    \(0\le c_1\le\cdots\le c_k\le \frac12\).
    \item[(ii)] If \(n=2k\), then the agents can be indexed as
    \(-k,\ldots,-1,1,\ldots,k\) so that
    \[
        b_{-i}=-a_i,\qquad b_i=c_i,\qquad i\in[k],
    \]
    where \(0\le a_1\le\cdots\le a_k\le \frac12\) and
    \(0\le c_1\le\cdots\le c_k\le \frac12\).
\end{itemize}
    In both cases,
    \(
        \mathrm{OPT}(\mathbf{b})=\mathrm{SC}(\mathbf{b},0)=\sum_{i=1}^k(a_i+c_i).
    \)
\end{lemma}

\begin{proof}
    Let \(o\in G\) be an optimal location for \(\mathbf{x}\). By neutrality,
    rotating the circle so that \(o\) becomes \(0\) preserves both the mechanism's
    social cost and the optimal social cost. Thus the approximation ratio is
    unchanged. We henceforth assume that \(0\) is an optimal location.

    Let $P,N,Z,H$ denote the numbers of agents in $(0,\frac12)$, in $(-\frac12,0)$, at $0$, and at the antipodal point, respectively. The antipodal point is counted once, regardless of whether its representative is $-\frac12$ or $\frac12$. Thus $P+N+Z+H=n$.
    The one-sided rates of change of social cost when moving clockwise and counterclockwise away from $0$ are, respectively,
    \[
        N+Z-P-H\qquad\text{and}\qquad P+Z-N-H.
    \]
    In either direction, agents at $0$ become farther away and agents at the antipodal point become closer. These rates are derivatives at $0$, so no finite movement past another report is involved. Optimality of $0$ requires both rates to be nonnegative, giving
    \[
        P+H\le N+Z,\qquad N+H\le P+Z.
    \]
    Since the four counts sum to $n$, these inequalities imply $P,N\le\lfloor\frac n2\rfloor$. Adding them also gives $H\le Z$.

    If $n=2k+1$ and $Z=0$, then $H=0$ and the two inequalities force $P=N$, contradicting oddness of $n$. Hence $Z\ge1$, and we choose one agent at $0$ as the central agent. The two open sides contain $P$ and $N$ agents and have remaining capacities $k-P$ and $k-N$. These are nonnegative and sum to $Z+H-1$, exactly the number of remaining copies at $0$ or at the antipodal point. Assign these copies to fill the two capacities, using the appropriate representative of the antipodal point on each side.
    If $n=2k$, no central agent is removed. The capacities $k-P$ and $k-N$ instead sum to $Z+H$, so the same assignment yields exactly $k$ agents on each side.

    Finally, by anonymity, relabeling agents according to their clockwise order
    does not change the output of \(f\). Therefore the resulting
    profile has the claimed form and the same approximation ratio. Since
    distances from \(0\) to \(-a_i\) and \(c_i\) are \(a_i\) and \(c_i\),
    respectively, the optimal social cost is
    \(\mathrm{OPT}(\mathbf{b})=\sum_{i=1}^k(a_i+c_i)\).
\end{proof}

After normalization, we cut the circle at the antipodal point of \(0\), that is, the glued point of $-\frac12$ and $\frac12$. The new distance function is called the \emph{cut
metric},  which is the line metric \(d'(x,y)=|x-y|\) on the representatives in
\([-\frac12,\frac12]\). Compared to the original circle metric $d$, the cut metric can only increase distances. For any pair of points on the same side of \(0\), the increase is
zero. For a pair of points $(b_{-i},b_j)$ on different sides of $0$, where \(i,j\in[k]\), the increase is called the \emph{crossing saving}:
\begin{align}
    \sigma_{ij}
    =
    d'(-a_i,c_j)-d(-a_i,c_j)
    =
    \bigl(2(a_i+c_j)-1\bigr)_+ ,
    \label{eq:crossing-saving}
\end{align}
where \((z)_+=\max\{z,0\}\). The following lemma records how these quantities
relate the cut-metric social cost to the exact circle social cost.

\begin{lemma}[Cut-saving identity]\label{lem:cut-saving-identity}
    Fix a normalized profile $\mathbf b$. Let \(f\) be any peaks-only randomized mechanism.
    Let \(C_f\) be the exact circle social cost of \(f\), and \(C'_f\) the
    social cost computed with the cut metric \(d'\). Suppose that \(f\) selects
    the left report \(-a_i\) with probability \(p_i^-\), and the right report
    \(c_j\) with probability \(p_j^+\), for any $i,j\in[k]$. Then the cut saving of $f$ is
    \[
        C'_f-C_f
        =
        \sum_{i=1}^k p_i^-\sum_{j=1}^k\sigma_{ij}
        +
        \sum_{j=1}^k p_j^+\sum_{i=1}^k\sigma_{ij}.
    \]
\end{lemma}

\begin{proof}
    Since \(f\) is peaks-only, its realized outcome is one of the reported
    locations. Fix first the event that the mechanism chooses \(-a_i\). For another
    agent also located on the left side, the cut distance and the circle distance
    to \(-a_i\) are the same. The same is true for any agent located at \(0\), if
    such an agent exists. The only distances that may change are those from
    \(-a_i\) to right-side reports \(c_j\). By
    \eqref{eq:crossing-saving}, the cut metric overestimates the distance from
    \(-a_i\) to \(c_j\) by exactly \(\sigma_{ij}\). Hence, conditional on
    choosing \(-a_i\), the total increase in social cost is
    \(\sum_{j=1}^k\sigma_{ij}\).

    The argument for choosing a right-side report \(c_j\) is symmetric: the total
    increase is \(\sum_{i=1}^k\sigma_{ij}\). If the mechanism chooses \(0\), the
    increase is zero because the distance from \(0\) to every report is the same
    under \(d\) and \(d'\). Taking expectation over the mechanism's lottery gives
    the claimed identity.
\end{proof}

We use \(C_{\mathrm{RD}}\) and \(C_{\mathrm{PCD}}\) for the exact circle social costs of \(\mathrm{RD}\) and
\(\mathrm{PCD}\), and \(C'_{\mathrm{RD}}\) and \(C'_{\mathrm{PCD}}\) for the corresponding social costs
computed under the cut metric \(d'\). The cut-saving identity for \(\mathrm{RD}\) has a particularly simple form. In the odd case
\(n=2k+1\), every reported copy is selected with probability \(1/(2k+1)\), and
there is one copy at \(0\). Thus the cut saving is
\[
    C'_{\mathrm{RD}}-C_{\mathrm{RD}}
    =
    \frac{2}{2k+1}\sum_{i=1}^k\sum_{j=1}^k\sigma_{ij}.
\]
In the even case \(n=2k\), every reported copy is selected with probability
\(1/(2k)\), and hence
\[
    C'_{\mathrm{RD}}-C_{\mathrm{RD}}
    =
    \frac{1}{k}\sum_{i=1}^k\sum_{j=1}^k\sigma_{ij}.
\]

\begin{table}[H]
\centering
\caption{Frequently used notation in the upper-bound analysis.}
\label{tab:notation}
\begin{tabular}{p{0.22\textwidth}p{0.66\textwidth}}
\toprule
\textbf{Notation} & \textbf{Meaning} \\
\midrule





$\mathbf{b}$ & A normalized profile for which $0$ is an optimal facility location. \\

$a_i,c_i$ & Distances of reports from $0$ on the two sides of the cut: the normalized reports are $-a_i$ and $c_i$. \\

$\mathrm{OPT}$ & The optimal social cost of the normalized profile:
$\mathrm{OPT}(\mathbf{b})=\sum_{i=1}^k(a_i+c_i)$. \\

$d'$ & The cut metric obtained after cutting the circle at the point antipodal to $0$. \\

$C_f$, $C'_f$ & The social cost of mechanism $f$ under the true circle metric $d$, and under the cut metric $d'$, respectively. \\

$\sigma_{ij}$ & The saving from using the true circle metric instead of the cut metric for the crossing pair $(-a_i,c_j)$:
$\sigma_{ij}=\bigl(2(a_i+c_j)-1\bigr)_+$. \\

$x_i,y_i$ & Reversed-coordinate notation used in the proof:
$x_i=a_i$ and $y_i=c_{k+1-i}$. \\

$\alpha_i,\beta_i$ & Coordinate increments:
$\alpha_i=x_i-x_{i-1}$ and $\beta_i=y_i-y_{i+1}$, where $x_0=0$ and $y_{k+1}=0$. \\

$\tau_{ij}$ & Crossing saving in reversed coordinates:
$\tau_{ij}=\bigl(2(x_i+y_j)-1\bigr)_+$. \\


$F_{\mathrm{odd}}$, $F_{\mathrm{even}}$ & Cut-metric excess terms:
$F_{\mathrm{odd}}=C'_{\mathrm{RD}}+C'_{\mathrm{PCD}}-3\mathrm{OPT}$ and
$F_{\mathrm{even}}=C'_{\mathrm{RD}}+C'_{\mathrm{EPCD}}-3\mathrm{OPT}$. \\

\bottomrule
\end{tabular}
\end{table}

\subsection{Odd Number of Agents}\label{subsec:odd}

We first prove the $\frac32$-approximation for odd $n$. Throughout this subsection,
let \(n=2k+1\), and we consider a normalized profile \(b_{-i}=-a_i,\,\, b_0=0,\,\, b_i=c_i, \) for \(i\in[k]\),
where \(0\le a_1\le\cdots\le a_k\le 1/2\) and
\(0\le c_1\le\cdots\le c_k\le 1/2\). We also set \(a_0=c_0=0\), and write the
optimal social cost as
\(\mathrm{OPT}=\sum_{i=1}^k(a_i+c_i)\) according to Lemma~\ref{lem:normalization}.

Our goal is to prove
\[C_{\mathrm{RD}}+C_{\mathrm{PCD}}\le 3\mathrm{OPT}.\] Since the mechanism \(\frac12 \mathrm{RD}+\frac12 \mathrm{PCD}\)
averages the two costs, this will imply the desired ratio of \(\frac32\).
The proof proceeds in three steps. First, we work with the cut metric \(d'\),
which only overestimates the true circle metric, and compute the corresponding
cut-metric social costs \(C'_{\mathrm{RD}}\) and \(C'_{\mathrm{PCD}}\) in Lemma~\ref{lem:odd-cut-costs}. Second,
we define the \emph{cut excess}
\begin{align}
    F_{\mathrm{odd}}
    =
    C'_{\mathrm{RD}}+C'_{\mathrm{PCD}}-3\mathrm{OPT},
    \label{eq:odd-cut-excess-def}
\end{align}
which is the amount by which the total cut-metric social cost of the
two mechanisms exceeds \(3\mathrm{OPT}\),
and we rewrite it in increment form in Lemma~\ref{lem:odd-increment-form}.  Third, in Lemma~\ref{lem:odd-domination} we show that the cut
excess $F_{\mathrm{odd}}$ is no larger than the cut saving of \(\mathrm{PCD}\), namely
\(C'_{\mathrm{PCD}}-C_{\mathrm{PCD}}\). Since the cut saving of
\(\mathrm{RD}\) is nonnegative, these two facts together immediately imply
\(C_{\mathrm{RD}}+C_{\mathrm{PCD}}\le 3\mathrm{OPT}\).

Let the ordered reported copies in clockwise order be \(-a_k,\ldots,-a_1,0,c_1,\ldots,c_k\) as in Lemma~\ref{lem:normalization}.
From the definition of \(\mathrm{PCD}\), the probability assigned to a left-side report $b_{-i}$,
to the origin $0$, and to a right-side report $b_i$ is
\(p_i^- = c_{k+1-i}-c_{k-i}
, \,\, p_0 = 1-a_k-c_k, \,\,\text{and}\,\,  p_i^+ = a_{k+1-i}-a_{k-i},\)
respectively. These probabilities are just the opposing gaps in the
clockwise order.
The next lemma expresses the cut-metric social costs in terms of the coordinates
\(a_i\) and \(c_i\).

\begin{lemma}\label{lem:odd-cut-costs}
    For every normalized odd-agent profile, the cut-metric social costs satisfy
    \begin{equation*}
        \begin{gathered}
             C'_{\mathrm{RD}}
        =
        \sum_{i=1}^k \frac{4i}{2k+1}(a_i+c_i),\quad\text{and}\\
        C'_{\mathrm{PCD}}
        =
        \mathrm{OPT}+
        \sum_{r=1}^k(2r-1)
        \left(
            c_{k+1-r}(a_r-a_{r-1})
            +
            a_{k+1-r}(c_r-c_{r-1})
        \right).
        \end{gathered}
    \end{equation*}
\end{lemma}

\begin{proof}
    We first compute $C'_{\mathrm{RD}}$.  Let \(S'(y)\) denote the social cost of location \(y\) under the cut metric.
    Since the cut metric is the line metric on the representatives in
    \([-\frac12,\frac12]\), we have \(S'(0)=\mathrm{OPT}\). Fix an index \(i>0\). When the facility moves from
    \(c_{i-1}\) to \(c_i\), the distance to each of the \(k+i\) agents who are weakly to the left
    of \(c_{i-1}\) increases by \(c_i-c_{i-1}\), while the distance to each of the
    \(k-i+1\) agents who are weakly to the right of \(c_i\) decreases by \(c_i-c_{i-1}\).
    Therefore
    \[
        S'(c_i)-S'(c_{i-1})
        =
        (k+i)(c_i-c_{i-1})-(k-i+1)(c_i-c_{i-1})
        =
        (2i-1)(c_i-c_{i-1}).
    \]
    Summing these increments gives
    \begin{equation}\label{eq:sci}
        S'(c_i)-\mathrm{OPT}=
        \sum_{h=1}^i (2h-1)(c_h-c_{h-1})
        =
        (2i-1)c_i-2\sum_{h=1}^{i-1}c_h.
    \end{equation}
    By symmetry,
    \[
        S'(-a_i)-\mathrm{OPT}
        =
        (2i-1)a_i-2\sum_{h=1}^{i-1}a_h .
    \]
    Averaging these quantities over the \(2k+1\) equally likely dictators,
    including the dictator at \(0\), yields
    \[
    \begin{aligned}
        C'_{\mathrm{RD}}
        &=
        \frac{1}{2k+1}
        \left(
            S'(0)+\sum_{i=1}^k\bigl(S'(c_i)+S'(-a_i)\bigr)
        \right)  \\
        &=
        \mathrm{OPT}+
        \frac{1}{2k+1}
        \sum_{i=1}^k
        \bigl(S'(c_i)-\mathrm{OPT}+S'(-a_i)-\mathrm{OPT}\bigr)  \\
        &=
        \mathrm{OPT}+
        \frac{1}{2k+1}
        \sum_{i=1}^k
        \left(
            (2i-1)(a_i+c_i)
            -
            2\sum_{h=1}^{i-1}(a_h+c_h)
        \right).
    \end{aligned}
    \]
    We now collect the coefficient of each \(a_h+c_h\) in the last summation. The
    term \(a_h+c_h\) appears once in the first part with coefficient \(2h-1\).
    It also appears in the double sum \(\sum_{i=1}^k\sum_{h=1}^{i-1}(a_h+c_h)\)
    exactly for those \(i\) satisfying \(i>h\), namely for
    \(i=h+1,\ldots,k\). Thus it appears \(k-h\) times in the double sum, with
    coefficient \(-2\) each time. Therefore the total coefficient of
    \(a_h+c_h\) is \( (2h-1)-2(k-h)=4h-2k-1\).
    Hence
    \[
    \begin{aligned}
        C'_{\mathrm{RD}}
        &=
        \mathrm{OPT}+
        \frac{1}{2k+1}
        \sum_{h=1}^k
        (4h-2k-1)(a_h+c_h)  \\
        &=
        \frac{1}{2k+1}
        \sum_{h=1}^k
        \bigl((2k+1)+(4h-2k-1)\bigr)(a_h+c_h)  \\
        &=
        \sum_{h=1}^k
        \frac{4h}{2k+1}(a_h+c_h).
    \end{aligned}
    \]

    We now compute \(C'_{\mathrm{PCD}}\). Conditional on choosing \(0\), the cut-metric social
    cost is \(\mathrm{OPT}\). Therefore
    \[
        C'_{\mathrm{PCD}}
        =
        \mathrm{OPT}+
        \sum_{i=1}^k p_i^-\bigl(S'(-a_i)-\mathrm{OPT}\bigr)
        +
        \sum_{i=1}^k p_i^+\bigl(S'(c_i)-\mathrm{OPT}\bigr).
    \]
    For the right-side reports, recall that \(p_i^+=a_{k+1-i}-a_{k-i}\).
    Using the telescoping representation of \eqref{eq:sci},
  the right-side contribution is
    \[
    \begin{aligned}
        \sum_{i=1}^k p_i^+\bigl(S'(c_i)-\mathrm{OPT}\bigr)
        &=
        \sum_{i=1}^k
        (a_{k+1-i}-a_{k-i})
        \sum_{h=1}^i (2h-1)(c_h-c_{h-1})  \\
        &=
        \sum_{i=1}^k
        \sum_{h=1}^i
        (a_{k+1-i}-a_{k-i})(2h-1)(c_h-c_{h-1})  \\
        &=
        \sum_{h=1}^k
        (2h-1)(c_h-c_{h-1})
        \sum_{i=h}^k(a_{k+1-i}-a_{k-i})  \\
        &=
        \sum_{h=1}^k
        (2h-1)a_{k+1-h}(c_h-c_{h-1}).
    \end{aligned}    \]
   In the same way, the left-side 
    contribution is
    \[
    \sum_{h=1}^k(2h-1)c_{k+1-h}(a_h-a_{h-1}).
    \]
    Combining the two parts gives
    the claimed formula for \(C'_{\mathrm{PCD}}\).
\end{proof}

Now we consider the odd-case cut excess $F_{\mathrm{odd}}$ defined in \eqref{eq:odd-cut-excess-def}. 
By Lemma~\ref{lem:odd-cut-costs},
\begin{align}
    F_{\mathrm{odd}}
    &=  \sum_{r=1}^k(2r-1)
    \left(
        c_{k+1-r}(a_r-a_{r-1})
        +
        a_{k+1-r}(c_r-c_{r-1})
    \right)                                      
    +
    \sum_{i=1}^k
    \frac{4i}{2k+1}(a_i+c_i)-2\mathrm{OPT}\nonumber \\
    &=    \sum_{r=1}^k(2r-1)
    \left(
        c_{k+1-r}(a_r-a_{r-1})
        +
        a_{k+1-r}(c_r-c_{r-1})
    \right)                                      
    -
    \sum_{i=1}^k
    \frac{2(2k+1-2i)}{2k+1}(a_i+c_i).
    \label{eq:odd-cut-excess}
\end{align}
The following lemma provides a decomposition of $F_{\mathrm{odd}}$ that separates the positive
bilinear part $B$ from a nonnegative linear term $R$, so that $F_{\mathrm{odd}}=B-R$. We call $R$ a \emph{linear budget} because we allocate it across the bilinear cells of $B$ to offset their contributions. Only the excess left after this allocation must be bounded by the crossing savings.

For this goal, it is convenient to reverse the right side of the profile. Define $x_i=a_i,\,\, y_i=c_{k+1-i},\,\, i=1,\ldots,k.$
Then \(0\le x_1\le\cdots\le x_k\le 1/2\), while
\(1/2\ge y_1\ge\cdots\ge y_k\ge 0\). Let $\alpha_i=x_i-x_{i-1},\,\,
    \beta_i=y_i-y_{i+1}$ be the gaps,
where \(x_0=0\) and \(y_{k+1}=0\). Thus \(\alpha_i,\beta_i\ge 0\). Finally set
\(\tau_{ij}=(2(x_i+y_j)-1)_+\). This is the same crossing-saving term as
\(\sigma_{ij}\), after reversing the right-side index. Since the  \(\mathrm{PCD}\) mechanism chooses the left report \(-x_i\) with probability
    \(y_i-y_{i+1}=\beta_i\), and chooses the right report \(y_j\) with
    probability \(x_j-x_{j-1}=\alpha_j\), according to
    Lemma~\ref{lem:cut-saving-identity},   the cut saving of \(\mathrm{PCD}\) is
    \begin{align}
       C_{\mathrm{PCD}}'-C_{\mathrm{PCD}}
        =
        \sum_{i=1}^k\beta_i\sum_{j=1}^k\tau_{ij}
        +
        \sum_{j=1}^k\alpha_j\sum_{i=1}^k\tau_{ij}.
        \label{eq:odd-pcd-saving}
    \end{align}

\begin{lemma}[Odd cut-excess decomposition]\label{lem:odd-increment-form}
   The cut excess $F_{\mathrm{odd}}$ in \eqref{eq:odd-cut-excess} is $F_{\mathrm{odd}}=B-R$, where
    \begin{align*}
        B
        =
        \sum_{1\le r\le s\le k}
        2(k-s+r)\alpha_r\beta_s\quad
       ~\text{~and~}~\quad
        R
        =
        \sum_{r=1}^k\frac{2(k+1-r)^2}{2k+1}\alpha_r
        +
        \sum_{s=1}^k\frac{2s^2}{2k+1}\beta_s.
    \end{align*}
\end{lemma}

\begin{proof}
    We first rewrite the bilinear part in \eqref{eq:odd-cut-excess}. Since
    \(c_{k+1-r}=y_r=\sum_{s=r}^k\beta_s\) and
    \(a_r-a_{r-1}=\alpha_r\), the first half contributes
    \(\sum_{1\le r\le s\le k}(2r-1)\alpha_r\beta_s\). For the second half, substitute
    \(s=k+1-r\). Then
    \(a_{k+1-r}=x_s\) and \(c_r-c_{r-1}=\beta_s\), while
    \(2r-1=2k+1-2s\). Since \(x_s=\sum_{r=1}^s\alpha_r\), the second half
    contributes \(\sum_{1\le r\le s\le k}(2k+1-2s)\alpha_r\beta_s\). Adding the two
    contributions gives the expression $B$.

    We now rewrite the linear term in \eqref{eq:odd-cut-excess}. For the
    \(a\)-coordinates, \(a_i=\sum_{r=1}^i\alpha_r\). Fixing $r$, the coefficient of
    \(\alpha_r\) in the expression is
    \[
        \sum_{i=r}^k \frac{2(2k+1-2i)}{2k+1}
        =
        \frac{2(k+1-r)^2}{2k+1}.
    \]
    For the \(c\)-coordinates and a fixed index $s$, reversing the order gives the coefficient
    \(\frac{2s^2}{2k+1}\) of \(\beta_s\). Therefore, the linear term in \eqref{eq:odd-cut-excess} is equal to $R$. This proves \(F_{\mathrm{odd}}=B-R\).
\end{proof}

The next step is to show that the cut excess $F_{\mathrm{odd}}$ is dominated by the cut saving of \(\mathrm{PCD}\).
The decomposition \(F_{\mathrm{odd}}=B-R\) reduces the task to
showing that the linear budget \(R\), together with the cut saving
\(C_{\mathrm{PCD}}'-C_{\mathrm{PCD}}\), dominates the bilinear term \(B\).

\begin{lemma}[Odd crossing domination]\label{lem:odd-domination}
    For every normalized odd-agent profile,
    \[
        F_{\mathrm{odd}}\le C_{\mathrm{PCD}}'-C_{\mathrm{PCD}}.
    \]
\end{lemma}

\begin{proof}
    Let \(n=2k+1\). We use the decomposition
    \(F_{\mathrm{odd}}=B-R\) from Lemma~\ref{lem:odd-increment-form}.
    The proof compares each cell of \(B\) with an allocated part of the linear
    term \(R\). Whenever this allocation is not sufficient, the remaining error
    is paid for by crossing savings.

    Fix a cell \(1\le r\le s\le k\). Set
    \(X=x_s\), \(Y=y_r\), \(S=X+Y\), and
    \(m=k+1-r+s\). Then \(n-m=k-s+r\), and
    \(m\ge k+1>n/2\). If \(\alpha_r\beta_s=0\), this cell contributes nothing.
    Otherwise, \(X>0\) and \(Y>0\), because
    \(X=x_s\ge \sum_{h=1}^s\alpha_h\ge \alpha_r>0\), and
    \(Y=y_r\ge \sum_{h=s}^k\beta_h\ge \beta_s>0\).
    
    By Cauchy's inequality in the form
    \((u^2/A+v^2/B)(A+B)\ge (u+v)^2\), applied with
    \(u=k+1-r\), \(v=s\), \(A=Y\), and \(B=X\), we have
    \[
        \frac{(k+1-r)^2}{Y}+\frac{s^2}{X}
        \ge
        \frac{m^2}{S}.
    \]
    
    We now explain how this inequality is used. For the cell \((r,s)\), the
    contribution of \(B\) is \(2(n-m)\alpha_r\beta_s\). We allocate to this cell
    the following part of the linear term \(R\):
    \[
        A_{rs}
        =
        \frac{2\alpha_r\beta_s}{n}
        \left(
            \frac{(k+1-r)^2}{Y}
            +
            \frac{s^2}{X}
        \right).
    \]
    By the preceding Cauchy inequality,
    \[
        A_{rs}
        \ge
        \frac{2m^2}{nS}\alpha_r\beta_s .
    \]
    Thus \(A_{rs}\) may or may not already cover the cell contribution
    \(2(n-m)\alpha_r\beta_s\). We define a conservative upper bound on the possible
    uncovered part by
    \[
        E_{rs}
        =
        \left(
            2(n-m)-\frac{2m^2}{nS}
        \right)_+
        \alpha_r\beta_s .
    \]
    Then
    \begin{align}
        2(n-m)\alpha_r\beta_s
        &\le
        A_{rs}+E_{rs} =
        \frac{2\alpha_r\beta_s}{n}
        \left(
            \frac{(k+1-r)^2}{Y}
            +
            \frac{s^2}{X}
        \right)
        +
        E_{rs}.
        \label{eq:odd-cell-linear-allocation}
    \end{align}
    Indeed, if \(2(n-m)\le 2m^2/(nS)\), then \(E_{rs}=0\) and
    \(A_{rs}\ge 2(n-m)\alpha_r\beta_s\). If
    \(2(n-m)>2m^2/(nS)\), then
    \(A_{rs}\ge 2m^2\alpha_r\beta_s/(nS)\), and adding \(E_{rs}\) gives at least
    \(2(n-m)\alpha_r\beta_s\).
    
    Summing the allocated linear terms in
    \eqref{eq:odd-cell-linear-allocation} over all cells gives at most \(R\).
    Indeed, for a fixed \(r\), since \(Y=y_r=\sum_{s=r}^k\beta_s\),
    \[
        \sum_{s=r}^k
        \frac{\beta_s}{y_r}
        =
        1.
    \]
    Thus the first allocated part contributes
    \(\sum_r\frac{2(k+1-r)^2}{n}\alpha_r\) in total. Similarly,
    \(X=x_s=\sum_{r=1}^s\alpha_r\), so the second allocated part contributes
    \(\sum_s\frac{2s^2}{n}\beta_s\). The sum of these two-part contributions is exactly $R$. Therefore
    \begin{align}
        F_{\mathrm{odd}}
        =
        B-R
        \le
        \sum_{1\le r\le s\le k}E_{rs}.
        \label{eq:odd-excess-error-sum}
    \end{align}

    It remains to show that the total error $\sum E_{rs}$ is at most \(C_{\mathrm{PCD}}'-C_{\mathrm{PCD}}\). Let
    \(z_{rs}=\tau_{sr}=(2(x_s+y_r)-1)_+=(2S-1)_+\). We first prove the cellwise bound
    \begin{align}
        E_{rs}
        \le
        \alpha_r\beta_s
        \left(
            \frac{k-s+1}{y_r}
            +
            \frac{r}{x_s}
        \right)
        z_{rs}.
        \label{eq:odd-cell-error-saving}
    \end{align}
    If the coefficient defining \(E_{rs}\) is non-positive, this is immediate.
    Otherwise,
    \(2(n-m)-2m^2/(nS)>0\). Since \(m>n/2\), this implies \(S>1/2\), and hence
    \(z_{rs}=2S-1\). Put \(u=k-s+1\) and \(v=r\). Then
    \(u+v=n+1-m\) and \(n-m=u+v-1\). Another application of Cauchy's
    inequality gives
    \[
        \frac{u}{Y}+\frac{v}{X}
        \ge
        \frac{(\sqrt{u}+\sqrt{v})^2}{S}.
    \]
    Thus it suffices to prove
    \[
        2(n-m)-\frac{2m^2}{nS}
        \le
        \frac{(\sqrt{u}+\sqrt{v})^2}{S}(2S-1).
    \]
    Multiplying by \(S\), the right-hand side minus the left-hand side equals
    \[
        2S-(u+v)+(4S-2)\sqrt{uv}+\frac{2m^2}{n}.
    \]
    Since \(S>1/2\), the middle term is non-negative and \(2S\ge 1\). Therefore
    the above expression is at least
    \[
        1-(u+v)+\frac{2m^2}{n}
        =
        m-n+\frac{2m^2}{n}
        =
        \frac{2m^2+mn-n^2}{n},
    \]
    which is non-negative because \(m\ge n/2\). This proves
    \eqref{eq:odd-cell-error-saving}.

    We now sum up \eqref{eq:odd-cell-error-saving} for all $r,s$ with $1\le r\le s\le k$. For fixed \(r\) and \(s\),
    monotonicity of the \(x\)-sequence implies
    \(\tau_{ir}\ge \tau_{sr}=z_{rs}\) for every \(i\ge s\). Hence
    \[
        (k-s+1)z_{rs}
        \le
        \sum_{i=1}^k\tau_{ir}.
    \]
    Using again \(\sum_{s=r}^k\beta_s/y_r=1\), we get, for every fixed \(r\),
    \[
        \sum_{s=r}^k
        \alpha_r\beta_s
        \frac{k-s+1}{y_r}
        z_{rs}
        \le
        \alpha_r\sum_{i=1}^k\tau_{ir}.
    \]
    Similarly, for fixed \(s\), monotonicity of the \(y\)-sequence implies
    \(\tau_{sj}\ge \tau_{sr}=z_{rs}\) for every \(j\le r\), and therefore
    \(r z_{rs}\le \sum_{j=1}^k\tau_{sj}\). Since
    \(\sum_{r=1}^s\alpha_r/x_s=1\), we obtain
    \[
        \sum_{r=1}^s
        \alpha_r\beta_s
        \frac{r}{x_s}
        z_{rs}
        \le
        \beta_s\sum_{j=1}^k\tau_{sj}.
    \]
    Summing the last two inequalities over all \(r\) and \(s\), and using
    \eqref{eq:odd-pcd-saving}, gives
    \[
        \sum_{1\le r\le s\le k}E_{rs}
        \le
        C_{\mathrm{PCD}}'-C_{\mathrm{PCD}}.
    \]
    Together with \eqref{eq:odd-excess-error-sum}, this proves the lemma.
\end{proof}

Now we are ready to derive the $\frac32$-approximation for any odd number of agents.

\begin{proof}[Proof of Theorem~\ref{thm:main-upper} in the odd case] Consider any normalized profile $\mathbf b$.
    By definition of \(F_{\mathrm{odd}}\),
    \(C'_{\mathrm{RD}}+C'_{\mathrm{PCD}}=3\mathrm{OPT}+F_{\mathrm{odd}}\). 
    Hence
    \[
        C_{\mathrm{RD}}+C_{\mathrm{PCD}}-3\mathrm{OPT}
        =
        F_{\mathrm{odd}}
        -
        (C'_{\mathrm{RD}}-C_{\mathrm{RD}})
        -
        (C_{\mathrm{PCD}}'-C_{\mathrm{PCD}})
        \le
        F_{\mathrm{odd}}-(C_{\mathrm{PCD}}'-C_{\mathrm{PCD}}),
    \]
    where the inequality is due to the facts that the cut saving of any mechanism is non-negative and thus $C'_{\mathrm{RD}}-C_{\mathrm{RD}}\ge 0$.
   Since Lemma~\ref{lem:odd-domination} shows that $F_{\mathrm{odd}}\le C_{\mathrm{PCD}}'-C_{\mathrm{PCD}}$, we obtain \[C_{\mathrm{RD}}+C_{\mathrm{PCD}}\le 3\mathrm{OPT}.\] The social cost of our mechanism
    \(\frac12 \mathrm{RD}+\frac12 \mathrm{PCD}\) is
    \(\frac12(C_{\mathrm{RD}}+C_{\mathrm{PCD}})\) in expectation, which is at most \(\frac32\mathrm{OPT}\). This gives the claimed approximation ratio. 
\end{proof}

\subsection{Even Number of Agents}\label{subsec:even}

 Let \(n=2k\ge 4\), and
consider a normalized profile $\mathbf b$ with \(b_{-i}=-a_i,\,\, b_i=c_i,\,\, \forall i\in[k]\),
where \(0\le a_1\le\cdots\le a_k\le 1/2\) and
\(0\le c_1\le\cdots\le c_k\le 1/2\). By Lemma~\ref{lem:normalization}, the optimal (circle-metric) social cost is
\(\mathrm{OPT}=\sum_{i=1}^k(a_i+c_i)\).

We use the same notations $x_i=a_i,\,\,y_i=c_{k+1-i},\,\,\alpha_i=x_i-x_{i-1},\,\,\beta_i=y_i-y_{i+1}, \,\,\text{and}\,\, \tau_{ij}=(2(x_i+y_j)-1)_+$ as in Section~\ref{subsec:odd}. The ordered reported copies in clockwise order are $-x_k,\ldots,-x_1,y_k,\ldots,y_1$. We further define \(h=1-x_k-y_1\) as the wrap-around gap, which is non-negative.

As in the odd-agent case, our goal is again to prove
\[
    C_{\mathrm{RD}}+C_{\mathrm{EPCD}}\le 3\mathrm{OPT}.
\]
The proof follows in a similar way. We first pass to the cut metric and define the  cut excess as
\begin{equation}\label{eq:fio}
F_{\mathrm{even}}=C'_{\mathrm{RD}}+C'_{\mathrm{EPCD}}-3\mathrm{OPT}.
\end{equation}
We then rewrite this excess in increment form and show that it is dominated by the crossing saving of \(\mathrm{EPCD}\), namely \(C'_{\mathrm{EPCD}}-C_{\mathrm{EPCD}}\). Since the cut saving of \(\mathrm{RD}\) is nonnegative, this domination implies \(C_{\mathrm{RD}}+C_{\mathrm{EPCD}}\le 3\mathrm{OPT}\) and thus a $\frac32$-approximation. 
However, the even case is more involved because \(\mathrm{EPCD}\) assigns to each report the average of two adjacent opposing gaps, rather than a single opposing gap as in \(\mathrm{PCD}\). This averaging creates additional boundary terms near the wrap-around gap and splits each interior contribution across two neighboring cells. We handle these terms by first separating the harmless boundary part and then charging the remaining interior excess to the crossing savings of \(\mathrm{EPCD}\).



The formula \eqref{eq:epqq} for the \(\mathrm{EPCD}\) mechanism gives the following
probabilities in the normalized coordinates. Let \(p_i^-\) be the probability
of choosing the left report \(-x_i\), and \(q_j\) the probability of
choosing the right report \(y_j\). Then
\[
    p_1^-=\frac12(h+\beta_1),\qquad
    p_i^-=\frac12(\beta_{i-1}+\beta_i),\quad
     2\le i\le k-1,\qquad
    p_k^-=\frac12(\beta_{k-1}+\beta_k+\alpha_1),
\]
\[
    q_1=\frac12(\alpha_1+\beta_k+\alpha_2),\qquad
    q_j=\frac12(\alpha_j+\alpha_{j+1}),\quad
   2\le j\le k-1,\qquad
    q_k=\frac12(\alpha_k+h).
\]
Here the endpoint formulas simply reflect the two gaps adjacent to the cut point and
the gap containing the origin. Applying Lemma~\ref{lem:cut-saving-identity}, the
cut saving of \(\mathrm{EPCD}\) is
\begin{align}
   C'_{\mathrm{EPCD}}-C_{\mathrm{EPCD}}
    =
    \sum_{i=1}^k p_i^-\sum_{j=1}^k\tau_{ij}
    +
    \sum_{j=1}^k q_j\sum_{i=1}^k\tau_{ij}.
    \label{eq:even-epcd-saving}
\end{align}


\begin{lemma}[Even cut-excess decomposition]\label{lem:even-increment-form}
 The even-case cut excess $F_{\mathrm{even}}$ in \eqref{eq:fio} is
    \[
        F_{\mathrm{even}}=Q-L,\quad\text{where}
    \]
    \begin{equation*}
    \begin{gathered} 
        L
        =
        \sum_{i=1}^k\frac{(k+1-i)^2}{k}\alpha_i
        +
        \sum_{j=1}^k\frac{j^2}{k}\beta_j,\qquad\text{and} \\
        Q\!
        =\!
        \sum_{j=1}^{k-1}(k-j)(\alpha_1\beta_j+\beta_j\beta_k)
        \!+\!
        \sum_{i=2}^k(i\!-\!1)(\alpha_i\beta_k\!+\!\alpha_1\alpha_i)      +
        \sum_{i=2}^k k\alpha_i\beta_{i-1}
        +\!\!
        \sum_{2\le i\le j\le k-1}\!\!\!2(k-j+i-1)\alpha_i\beta_j.
        \end{gathered} 
    \end{equation*}
\end{lemma}

\begin{proof}
    We first compute the cut-metric social cost of \(\mathrm{RD}\), namely, $C_{\mathrm{RD}}'$. Most computations are the same as in the odd case in the proof of Lemma~\ref{lem:odd-cut-costs}. Let \(S'(y)\) be the social
    cost of location \(y\) under the cut metric. Since the cut metric is the line
    metric on the representatives in \([-\frac12,\frac12]\), we have \(S'(0)=\mathrm{OPT}\).
    
    Fix \(i\in\{1,\ldots,k\}\), and let \(\Delta_i=c_i-c_{i-1}\), where \(c_0=0\).
    When the facility moves from \(c_{i-1}\) to \(c_i\), the distance to the
    \(k+i-1\) agents weakly to the left of \(c_{i-1}\) increases by \(\Delta_i\),
    while the distance to the \(k-i+1\) agents weakly to the right of \(c_i\)
    decreases by \(\Delta_i\). Therefore
    \[
        S'(c_i)-S'(c_{i-1})
        =
        (k+i-1)\Delta_i-(k-i+1)\Delta_i
        =
        2(i-1)(c_i-c_{i-1}).
    \]
    Summing these increments from \(1\) to \(i\), we get
    \[
        S'(c_i)-\mathrm{OPT}
        =
        \sum_{h=1}^i 2(h-1)(c_h-c_{h-1})
        =
        2(i-1)c_i-2\sum_{h=1}^{i-1}c_h.
    \]
    The same argument on the left side gives
    \[
        S'(-a_i)-\mathrm{OPT}
        =
        \sum_{h=1}^i 2(h-1)(a_h-a_{h-1})
        =
        2(i-1)a_i-2\sum_{h=1}^{i-1}a_h.
    \]
    
    Now average over the \(2k\) possible dictators. We have
    \[
    \begin{aligned}
        C'_{\mathrm{RD}}
        &=
        \mathrm{OPT}+
        \frac1{2k}\sum_{i=1}^k
        \bigl(S'(-a_i)-\mathrm{OPT}+S'(c_i)-\mathrm{OPT}\bigr)  \\
        &=
        \mathrm{OPT}+
        \frac1{2k}\sum_{i=1}^k
        \left(
            2(i-1)(a_i+c_i)
            -
            2\sum_{h=1}^{i-1}(a_h+c_h)
        \right).
    \end{aligned}
    \]
    We rewrite the double sum by collecting the coefficient of each
    \(a_h+c_h\). The term \(a_h+c_h\) appears in
    \(\sum_{i=1}^k\sum_{h=1}^{i-1}(a_h+c_h)\) exactly for
    \(i=h+1,\ldots,k\), hence \(k-h\) times. Therefore
    \[
    \begin{aligned}
        C'_{\mathrm{RD}}
        &=
        \mathrm{OPT}+
        \frac1{2k}\sum_{h=1}^k
        \bigl(2(h-1)-2(k-h)\bigr)(a_h+c_h)  =
        \mathrm{OPT}+
        \sum_{h=1}^k
        \frac{2h-k-1}{k}(a_h+c_h)  \\
        &=
        \sum_{h=1}^k
        \frac{2h-1}{k}(a_h+c_h).
    \end{aligned}
    \]
    Thus
    \[
        C'_{\mathrm{RD}}-\mathrm{OPT}
        =
        \sum_{h=1}^k
        \frac{2h-k-1}{k}(a_h+c_h).
    \]

    We now expand \(C'_{\mathrm{EPCD}}\). It is useful to first record the cut-metric
    social costs relative to \(\mathrm{OPT}\). From the increment computation above,
    \[
        S'(-x_i)-\mathrm{OPT}
        =
        \sum_{r=1}^i 2(r-1)\alpha_r,
    \]
    and
    \[
        S'(y_j)-\mathrm{OPT}
        =
        \sum_{s=j}^k 2(k-s)\beta_s.
    \]
    The second formula follows because \(y_j\) is the \((k+1-j)\)-th right-side
    report in the original increasing order.
    
    Since \(\mathrm{EPCD}\) chooses a left report \(-x_i\) with probability \(p_i^-\) and a
    right report \(y_j\) with probability \(q_j\), and since the probabilities sum
    to one, we have
    \[
        C'_{\mathrm{EPCD}}-\mathrm{OPT}
        =
        \sum_{i=1}^k p_i^-\bigl(S'(-x_i)-\mathrm{OPT}\bigr)
        +
        \sum_{j=1}^k q_j\bigl(S'(y_j)-\mathrm{OPT}\bigr).
    \]
 We next expand the right-hand side and verify that it is exactly the
bilinear term \(Q\).   We compute the two terms separately.
    
    For the left reports, the coefficient of \(\alpha_r\) in \(S'(-x_i)-\mathrm{OPT}\) appears
    exactly when \(i\ge r\). Since the coefficient for \(r=1\) is zero, we get
    \[
    \begin{aligned}
        \sum_{i=1}^k p_i^-\bigl(S'(-x_i)-\mathrm{OPT}\bigr)
        &=
        \sum_{i=2}^k p_i^-\sum_{r=1}^i 2(r-1)\alpha_r
        =
        \sum_{r=2}^k 2(r-1)\alpha_r\sum_{i=r}^k p_i^- .
    \end{aligned}
    \]
    For \(r=2,\ldots,k\), using the displayed formulas for the probabilities
    \(p_i^-\), the tail probability is
    \[
        \sum_{i=r}^k p_i^-
        =
        \sum_{s=r}^{k-1}\beta_s
        +
        \frac12(\beta_{r-1}+\beta_k+\alpha_1).
    \]
    Therefore
    \[
    \begin{aligned}
        \sum_{i=1}^k p_i^-\bigl(S'(-x_i)-\mathrm{OPT}\bigr)
        &=
        \sum_{r=2}^k
        2(r-1)\alpha_r\sum_{s=r}^{k-1}\beta_s  +
        \sum_{r=2}^k
        (r-1)\alpha_r(\beta_{r-1}+\beta_k+\alpha_1) \\
        &=
        \sum_{2\le r\le s\le k-1}
        2(r-1)\alpha_r\beta_s  +
        \sum_{r=2}^k
        (r-1)\alpha_r\beta_{r-1}
        +
        \sum_{r=2}^k
        (r-1)\alpha_r\beta_k
        +
        \sum_{r=2}^k
        (r-1)\alpha_1\alpha_r .
    \end{aligned}
    \]
    
    For the right reports, the coefficient of \(\beta_s\) in \(S'(y_j)-\mathrm{OPT}\) appears
    exactly when \(j\le s\). Since the coefficient for \(s=k\) is zero, we get
    \[
    \begin{aligned}
        \sum_{j=1}^k q_j\bigl(S'(y_j)-\mathrm{OPT}\bigr)
        &=
        \sum_{s=1}^{k-1}2(k-s)\beta_s\sum_{j=1}^s q_j .
    \end{aligned}
    \]
    For \(s=1,\ldots,k-1\), using the displayed formulas for the probabilities
    \(q_j\), the prefix probability is
    \[
        \sum_{j=1}^s q_j
        =
        \sum_{r=2}^{s}\alpha_r
        +
        \frac12(\alpha_1+\beta_k+\alpha_{s+1}).
    \]
    Therefore
    \[
    \begin{aligned}
        \sum_{j=1}^k q_j\bigl(S'(y_j)-\mathrm{OPT}\bigr)
        &=
        \sum_{s=1}^{k-1}
        2(k-s)\beta_s\sum_{r=2}^{s}\alpha_r  +
        \sum_{s=1}^{k-1}
        (k-s)\beta_s(\alpha_1+\beta_k+\alpha_{s+1}) \\
        &=
        \sum_{2\le r\le s\le k-1}
        2(k-s)\alpha_r\beta_s +
        \sum_{s=1}^{k-1}
        (k-s)\alpha_1\beta_s
        +
        \sum_{s=1}^{k-1}
        (k-s)\beta_s\beta_k
        +
        \sum_{s=1}^{k-1}
        (k-s)\alpha_{s+1}\beta_s .
    \end{aligned}
    \]
    
    Adding the left and right contributions gives \(C'_{\mathrm{EPCD}}-\mathrm{OPT}=Q\). Indeed, the
    terms involving \(\alpha_1\) and \(\beta_k\) give the two boundary strips
    \[
        \sum_{s=1}^{k-1}
        (k-s)\alpha_1\beta_s
        +
        \sum_{s=1}^{k-1}
        (k-s)\beta_s\beta_k=\sum_{j=1}^{k-1}(k-j)(\alpha_1\beta_j+\beta_j\beta_k)
    \]
    and
    \[
        \sum_{r=2}^k
        (r-1)\alpha_r\beta_k
        +
        \sum_{r=2}^k
        (r-1)\alpha_1\alpha_r=\sum_{i=2}^k(i-1)(\alpha_i\beta_k+\alpha_1\alpha_i).
    \]
    Also we have
    \[
        \begin{aligned}
            \sum_{i=2}^k (i-1)\alpha_i\beta_{i-1} + \sum_{s=1}^{k-1}(k-s)\alpha_{s+1}\beta_s = \sum_{i=2}^k (i-1)\alpha_i\beta_{i-1} + \sum_{i=2}^k(k-i+1)\alpha_i\beta_{i-1} = \sum_{i=2}^k k\alpha_i\beta_{i-1}.
        \end{aligned}
    \]
    Finally, for every interior pair \(2\le r\le s\le k-1\), the coefficient of
    \(\alpha_r\beta_s\) is
    \[
        2(r-1)+2(k-s)=2(k-s+r-1),
    \]
    which gives the last term of \(Q\). Hence \(C'_{\mathrm{EPCD}}-\mathrm{OPT}=Q\), with \(Q\) given in the claim of the lemma.
    
    It remains to identify the linear term \(L\). From the computation of \(\mathrm{RD}\)
    above, we have
    \[
        C'_{\mathrm{RD}}
        =
        \sum_{i=1}^k\frac{2i-1}{k}(a_i+c_i).
    \]
    Hence
    \[
        C'_{\mathrm{RD}}-\mathrm{OPT}
        =
        \sum_{i=1}^k\frac{2i-k-1}{k}a_i
        +
        \sum_{i=1}^k\frac{2i-k-1}{k}c_i.
    \]
    
    We first rewrite the left-side part in terms of the increments
    \(\alpha_r=x_r-x_{r-1}=a_r-a_{r-1}\). Since
    \(a_i=\sum_{r=1}^i\alpha_r\), we have
    \[
    \begin{aligned}
        \sum_{i=1}^k\frac{2i-k-1}{k}a_i
        &=
        \sum_{i=1}^k\frac{2i-k-1}{k}
        \sum_{r=1}^i\alpha_r =
        \sum_{r=1}^k
        \alpha_r
        \sum_{i=r}^k\frac{2i-k-1}{k}  =
        \sum_{r=1}^k
        \frac{(r-1)(k+1-r)}{k}\alpha_r .
    \end{aligned}
    \]
    
    We now rewrite the right-side part. Recall that \(y_j=c_{k+1-j}\), or
    equivalently \(c_i=y_{k+1-i}\). Therefore, changing variables from
    \(i\) to \(j=k+1-i\), we obtain
    \[
    \begin{aligned}
        \sum_{i=1}^k\frac{2i-k-1}{k}c_i
        &=
        \sum_{j=1}^k
        \frac{2(k+1-j)-k-1}{k}y_j  =
        \sum_{j=1}^k
        \frac{k+1-2j}{k}y_j .
    \end{aligned}
    \]
    Since \(\beta_s=y_s-y_{s+1}\) and \(y_{k+1}=0\), each \(y_j\) can be written as
    \[
        y_j=\sum_{s=j}^k\beta_s.
    \]
    Thus
    \[
    \begin{aligned}
        \sum_{j=1}^k
        \frac{k+1-2j}{k}y_j
        &=
        \sum_{j=1}^k
        \frac{k+1-2j}{k}
        \sum_{s=j}^k\beta_s =
        \sum_{s=1}^k
        \beta_s
        \sum_{j=1}^s\frac{k+1-2j}{k} =
        \sum_{s=1}^k
        \frac{s(k-s)}{k}\beta_s .
    \end{aligned}
    \]
    Combining the two sides, we get
    \[
        C'_{\mathrm{RD}}-\mathrm{OPT}
        =
        \sum_{i=1}^k
        \frac{(i-1)(k+1-i)}{k}\alpha_i
        +
        \sum_{j=1}^k
        \frac{j(k-j)}{k}\beta_j.
    \]
    
    On the other hand,
    \[
        \mathrm{OPT}
        =
        \sum_{i=1}^k a_i+\sum_{j=1}^k y_j.
    \]
    Using \(a_i=\sum_{r=1}^i\alpha_r\), the coefficient of \(\alpha_i\) in
    \(\sum_{r=1}^k a_r\) is \(k+1-i\). Using
    \(y_j=\sum_{s=j}^k\beta_s\), the coefficient of \(\beta_j\) in
    \(\sum_{r=1}^k y_r\) is \(j\). Hence
    \[
        \mathrm{OPT}
        =
        \sum_{i=1}^k(k+1-i)\alpha_i
        +
        \sum_{j=1}^k j\beta_j.
    \]
    Therefore
    \[
    \begin{aligned}
        \mathrm{OPT}-(C'_{\mathrm{RD}}-\mathrm{OPT})
        &=
        \sum_{i=1}^k
        \left(
            k+1-i-\frac{(i-1)(k+1-i)}{k}
        \right)\alpha_i +
        \sum_{j=1}^k
        \left(
            j-\frac{j(k-j)}{k}
        \right)\beta_j \\
        &=
        \sum_{i=1}^k\frac{(k+1-i)^2}{k}\alpha_i
        +
        \sum_{j=1}^k\frac{j^2}{k}\beta_j
        =
        L.
    \end{aligned}
    \]
    Combining \(C'_{\mathrm{EPCD}}-\mathrm{OPT}=Q\) and \(\mathrm{OPT}-(C'_{\mathrm{RD}}-\mathrm{OPT})=L\), we obtain
    \[
        F_{\mathrm{even}}
        =
        C'_{\mathrm{RD}}+C'_{\mathrm{EPCD}}-3\mathrm{OPT}
        =
        Q-L.
    \]
\end{proof}

We now separate from \(Q-L\) the boundary terms created by the gap containing
the origin. 
The next lemma shows that these boundary terms are non-positive. 
Thus the cut excess $F_{\mathrm{even}}$ is controlled by only the
interior part.

\begin{lemma}[Boundary reduction]\label{lem:even-boundary-reduction}
   The even-case cut excess $F_{\mathrm{even}}$ satisfies
    \[
        F_{\mathrm{even}}\le Q_{\mathrm{int}}-L_{\mathrm{int}}, \quad\text{where}
    \]
    \begin{equation*}
            Q_{\mathrm{int}}
        =
        \sum_{r=1}^{k-1}k\alpha_{r+1}\beta_r
        +
        \sum_{1\le r<s\le k-1}2(k-s+r)\alpha_{r+1}\beta_s,\quad
        L_{\mathrm{int}}
        =
        \sum_{r=1}^{k-1}\frac{(k-r)^2}{k}\alpha_{r+1}
        +
        \sum_{s=1}^{k-1}\frac{s^2}{k}\beta_s
    \end{equation*}
    denote the terms in $Q$ and $L$ that are independent of boundary increments $\alpha_1$ and $\beta_k$.
\end{lemma}

\begin{proof}
    We start from the decomposition \(F_{\mathrm{even}}=Q-L\) in
    Lemma~\ref{lem:even-increment-form}. The terms of \(Q\) involving
    \(\alpha_1\) or \(\beta_k\) are
    \[
        \sum_{j=1}^{k-1}(k-j)(\alpha_1\beta_j+\beta_j\beta_k)
        +
        \sum_{i=2}^k(i-1)(\alpha_i\beta_k+\alpha_1\alpha_i).
    \]
   Let
\[
    \Gamma=\alpha_1+\beta_k.
\] These terms can be written as
    \[
        Q_{\partial}
        =
        \Gamma
        \left(
            \sum_{j=1}^{k-1}(k-j)\beta_j
            +
            \sum_{i=2}^k(i-1)\alpha_i
        \right).
    \]
    The remaining terms of \(Q\) are
    \[
        \sum_{i=2}^k k\alpha_i\beta_{i-1}
        +
        \sum_{2\le i\le j\le k-1}2(k-j+i-1)\alpha_i\beta_j .
    \]
    Reindexing the first coordinate by \(r=i-1\), this becomes exactly
    \[
        Q_{\mathrm{int}}
        =
        \sum_{r=1}^{k-1}k\alpha_{r+1}\beta_r
        +
        \sum_{1\le r<s\le k-1}2(k-s+r)\alpha_{r+1}\beta_s .
    \]
    Hence
    \[
        Q=Q_{\partial}+Q_{\mathrm{int}}.
    \]

    Next consider the linear term
    \[
        L
        =
        \sum_{i=1}^k\frac{(k+1-i)^2}{k}\alpha_i
        +
        \sum_{j=1}^k\frac{j^2}{k}\beta_j.
    \]
    The coefficient of \(\alpha_1\) is \(k\), and the coefficient of
    \(\beta_k\) is also \(k\). Therefore the boundary contribution to \(L\) is
    \(k(\alpha_1+\beta_k)=k\Gamma\). The remaining contribution is
    \[
        L_{\mathrm{int}}
        =
        \sum_{r=1}^{k-1}\frac{(k-r)^2}{k}\alpha_{r+1}
        +
        \sum_{s=1}^{k-1}\frac{s^2}{k}\beta_s .
    \]
    Thus
    \[
        L=k\Gamma+L_{\mathrm{int}}.
    \]

    Combining the two decompositions gives
    \[
        F_{\mathrm{even}}
        =
        Q-L
        =
        (Q_{\partial}-k\Gamma)
        +
        (Q_{\mathrm{int}}-L_{\mathrm{int}}).
    \]
    It remains to show that \(Q_{\partial}-k\Gamma\le 0\). Since
    \[
        \sum_{j=1}^{k-1}\beta_j
        =
        y_1-y_k
        \le y_1
        \le \frac12
    \]
    and
    \[
        \sum_{i=2}^k\alpha_i
        =
        x_k-x_1
        \le x_k
        \le \frac12,
    \]
    we have
    \[
        \sum_{j=1}^{k-1}(k-j)\beta_j
        +
        \sum_{i=2}^k(i-1)\alpha_i
        \le
        (k-1)\sum_{j=1}^{k-1}\beta_j
        +
        (k-1)\sum_{i=2}^k\alpha_i \le k-1
        < k .
    \]
    Multiplying by \(\Gamma\ge 0\), we obtain \(Q_{\partial}\le k\Gamma\).
    Therefore
    \[
        F_{\mathrm{even}}
        \le
        Q_{\mathrm{int}}-L_{\mathrm{int}},
    \]
    as claimed.
\end{proof}

The next step is to show that the cut excess $F_{\mathrm{even}}$  is dominated by the cut saving of $\mathrm{EPCD}$.

\begin{lemma}[Even crossing domination]\label{lem:even-domination}
    For every normalized even-agent profile,
    \[
        F_{\mathrm{even}}\le Q_{\mathrm{int}}-L_{\mathrm{int}}\le C'_{\mathrm{EPCD}}-C_{\mathrm{EPCD}}.
    \]
\end{lemma}

\begin{proof}
    By Lemma~\ref{lem:even-boundary-reduction}, it is enough to prove
    \[
        Q_{\mathrm{int}}-L_{\mathrm{int}}
        \le
        C'_{\mathrm{EPCD}}-C_{\mathrm{EPCD}}.
    \]
    Recall that
      \[
        Q_{\mathrm{int}}
        =
        \sum_{1\le r\le s\le k-1}(k-s+r)\alpha_{r+1}\beta_s
        +
        \sum_{1\le r<s\le k-1}(k-s+r)\alpha_{r+1}\beta_s.
    \]
    and
    \[
        L_{\mathrm{int}}
        =
        \sum_{r=1}^{k-1}\frac{(k-r)^2}{k}\alpha_{r+1}
        +
        \sum_{s=1}^{k-1}\frac{s^2}{k}\beta_s .
    \]

    We view each product \(\alpha_{r+1}\beta_s\), with
    \(1\le r\le s\le k-1\), as a cell. A diagonal cell \(r=s\) appears once in
    \(Q_{\mathrm{int}}\), with coefficient \(k\). An off-diagonal cell \(r<s\)
    appears twice, each time with coefficient \(k-s+r\). We call these two
    copies the \emph{northeast} copy and the \emph{southwest} copy, respectively.

    The proof assigns to each copy a portion of the linear budget $L_{\mathrm{int}}$. If this assigned portion is not
enough to cover the copy’s contribution, the remaining error is paid for by crossing savings. We will
show that the total assigned linear budget is at most $L_{\mathrm{int}}$, and that the total remaining error is at most $C'_{\mathrm{EPCD}}-C_{\mathrm{EPCD}}$. 

    For a northeast copy of cell \((r,s)\), where \(1\le r\le s\le k-1\), set
    \[
        X=x_{s+1},\qquad Y=y_r,\qquad Z=\tau_{s+1,r}.
    \]
    For a southwest copy of cell \((r,s)\), where \(1\le r<s\le k-1\), set
    \[
        X=x_s,\qquad Y=y_{r+1},\qquad Z=\tau_{s,r+1}.
    \]
    The following argument applies to both types of copies. Put
    \[
        S=X+Y,\qquad m=k-r+s.
    \]
    Then \(k-s+r=2k-m\) and \(m\ge k\). Cells with
    \(\alpha_{r+1}\beta_s=0\) have zero contribution, so we ignore them. For
    the remaining cells, the associated \(X\) and \(Y\) are positive.

    By Cauchy's inequality,
    \[
        \frac{(k-r)^2}{Y}+\frac{s^2}{X}
        \ge
        \frac{(k-r+s)^2}{X+Y}
        =
        \frac{m^2}{S}.
    \]
    Assign to this copy the linear budget
    \[
        \mathcal A
        =
        \frac{\alpha_{r+1}\beta_s}{2k}
        \left(
            \frac{(k-r)^2}{Y}
            +
            \frac{s^2}{X}
        \right).
    \]
    Then
    \[
        \mathcal A
        \ge
        \frac{m^2}{2kS}\alpha_{r+1}\beta_s.
    \]
    Since this copy contributes \((2k-m)\alpha_{r+1}\beta_s\) to $Q_{\mathrm{int}}$, the possible
    uncovered part is at most
    \begin{equation}\label{eq:error}
        E
        =
        \left(
            (2k-m)-\frac{m^2}{2kS}
        \right)_+
        \alpha_{r+1}\beta_s.
    \end{equation}
    Hence every copy satisfies
    \[
        (k-s+r)\alpha_{r+1}\beta_s
        \le
        \mathcal A+E.
    \]

    We now bound the error \(E\) by the local crossing saving \(Z=(2S-1)_+\). We claim
    that
    \[
        E
        \le
        \frac12\alpha_{r+1}\beta_s
        \left(
            \frac{k-s}{Y}
            +
            \frac{r}{X}
        \right)
        (2S-1)_+ .
    \]
    If the coefficient defining \(E\) is non-positive, this is immediate.
    Otherwise,
    \[
        (2k-m)-\frac{m^2}{2kS}>0.
    \]
    Since \(m\ge k\), this implies \(S>1/2\), and thus
    \((2S-1)_+=2S-1\).

    Let \(u=k-s\) and \(v=r\). Then \(u+v=2k-m\). Applying Cauchy's inequality
    again gives
    \[
        \frac{u}{Y}+\frac{v}{X}
        \ge
        \frac{(\sqrt u+\sqrt v)^2}{X+Y}
        =
        \frac{u+v+2\sqrt{uv}}{S}.
    \]
    Therefore it is enough to prove
    \[
        (2k-m)-\frac{m^2}{2kS}
        \le
        \frac12\cdot
        \frac{u+v+2\sqrt{uv}}{S}
        (2S-1).
    \]
    Multiplying by \(2S\), this is equivalent to
    \[
        (2S-1)(u+v+2\sqrt{uv})
        \ge
        2S(u+v)-\frac{m^2}{k}.
    \]
    The left-hand side minus the right-hand side is
    \[
        -(u+v)+(4S-2)\sqrt{uv}+\frac{m^2}{k}.
    \]
    This is non-negative because \(S>1/2\) and
    \(
        \frac{m^2}{k}\ge 2k-m=u+v
    \)
    (where the inequality follows from \(m\ge k\)). This proves the claimed error
    bound.

    We next sum the allocated linear budgets. For northeast copies,
    \(X=x_{s+1}\) and \(Y=y_r\). The part involving \((k-r)^2\) is at most
    \[
    \begin{aligned}
        \sum_{r=1}^{k-1}\sum_{s=r}^{k-1}
        \frac{\alpha_{r+1}\beta_s}{2k}
        \frac{(k-r)^2}{y_r}
        &\le
        \sum_{r=1}^{k-1}
        \frac{(k-r)^2}{2k}\alpha_{r+1},
    \end{aligned}
    \]
    because \(\sum_{s=r}^{k-1}\beta_s\le y_r\). Similarly, the part involving
    \(s^2\) is at most
    \[
        \sum_{s=1}^{k-1}\frac{s^2}{2k}\beta_s,
    \]
    because \(\sum_{r=1}^s\alpha_{r+1}\le x_{s+1}\). Hence all northeast copies
    use at most one half of \(L_{\mathrm{int}}\).

    The southwest copies are analogous. For them \(X=x_s\) and \(Y=y_{r+1}\).
    Using
    \[
        \sum_{s=r+1}^{k-1}\beta_s\le y_{r+1},
        \qquad
        \sum_{r=1}^{s-1}\alpha_{r+1}\le x_s,
    \]
    their total allocation is also at most one half of \(L_{\mathrm{int}}\).
    Therefore all copies together use at most \(L_{\mathrm{int}}\), and so
    \[
        Q_{\mathrm{int}}-L_{\mathrm{int}}
        \le
        \sum_{1\le r\le s\le k-1}E^+_{rs}
        +
        \sum_{1\le r<s\le k-1}E^-_{rs},
    \]
    where \(E^+_{rs}\) and \(E^-_{rs}\) denote the northeast and southwest
    errors as in \eqref{eq:error}.

    We now charge these errors to the cut saving of \(\mathrm{EPCD}\). For a
    northeast copy, \(Z^+_{rs}=\tau_{s+1,r}\), and the error bound gives
    \[
        E^+_{rs}
        \le
        \frac12\alpha_{r+1}\beta_s
        \left(
            \frac{k-s}{y_r}
            +
            \frac{r}{x_{s+1}}
        \right)
        Z^+_{rs}.
    \]
    By monotonicity of \(\tau_{ij}\),
    \[
        (k-s)Z^+_{rs}\le \sum_{i=1}^k\tau_{i,r},
        \qquad
        rZ^+_{rs}\le \sum_{j=1}^k\tau_{s+1,j}.
    \]
    Summing over all northeast copies gives
    \[
        \sum_{1\le r\le s\le k-1}E^+_{rs}
        \le
        \sum_{r=1}^{k-1}\frac{\alpha_{r+1}}{2}
        \sum_{i=1}^k\tau_{i,r}
        +
        \sum_{s=1}^{k-1}\frac{\beta_s}{2}
        \sum_{j=1}^k\tau_{s+1,j}.
    \]

    Similarly, for a southwest copy, \(Z^-_{rs}=\tau_{s,r+1}\), and
    \[
        E^-_{rs}
        \le
        \frac12\alpha_{r+1}\beta_s
        \left(
            \frac{k-s}{y_{r+1}}
            +
            \frac{r}{x_s}
        \right)
        Z^-_{rs}.
    \]
    By monotonicity,
    \[
        (k-s)Z^-_{rs}\le \sum_{i=1}^k\tau_{i,r+1},
        \qquad
        rZ^-_{rs}\le \sum_{j=1}^k\tau_{s,j}.
    \]
    Summing over all southwest copies gives
    \[
        \sum_{1\le r<s\le k-1}E^-_{rs}
        \le
        \sum_{r=1}^{k-1}\frac{\alpha_{r+1}}{2}
        \sum_{i=1}^k\tau_{i,r+1}
        +
        \sum_{s=1}^{k-1}\frac{\beta_s}{2}
        \sum_{j=1}^k\tau_{s,j}.
    \]

    Combining the two estimates,
    \begin{align}
        Q_{\mathrm{int}}-L_{\mathrm{int}}
        \le{}&
        \sum_{r=1}^{k-1}\frac{\alpha_{r+1}}{2}
        \left(
            \sum_{i=1}^k\tau_{i,r}
            +
            \sum_{i=1}^k\tau_{i,r+1}
        \right) +
        \sum_{s=1}^{k-1}\frac{\beta_s}{2}
        \left(
            \sum_{j=1}^k\tau_{s,j}
            +
            \sum_{j=1}^k\tau_{s+1,j}
        \right).
        \label{eq:upper-qint-lint}
    \end{align}

    Finally, compare this with the cut saving of \(\mathrm{EPCD}\). Let
    \[
        R_i=\sum_{j=1}^k\tau_{ij},
        \qquad
        H_j=\sum_{i=1}^k\tau_{ij}.
    \]
  Thus \(R_i\) is the total crossing saving when the mechanism chooses the left
    report \(-x_i\), and \(H_j\) is the total crossing saving when it chooses the
    right report \(y_j\).   By \eqref{eq:even-epcd-saving},
    \[
        C'_{\mathrm{EPCD}}-C_{\mathrm{EPCD}}
        =
        \sum_{i=1}^k p_i^-R_i+\sum_{j=1}^k q_jH_j.
    \]
    From the formulas for the probabilities \(q_j\), each increment
    \(\alpha_{r+1}\), for \(r=1,\ldots,k-1\), appears with coefficient \(1/2\)
    in both \(q_r\) and \(q_{r+1}\). Hence
    \[
        \sum_{j=1}^k q_jH_j
        \ge
        \sum_{r=1}^{k-1}
        \frac{\alpha_{r+1}}{2}(H_r+H_{r+1}).
    \]
    Similarly, each increment \(\beta_s\), for \(s=1,\ldots,k-1\), appears with
    coefficient \(1/2\) in both \(p_s^-\) and \(p_{s+1}^-\). Hence
    \[
        \sum_{i=1}^k p_i^-R_i
        \ge
        \sum_{s=1}^{k-1}
        \frac{\beta_s}{2}(R_s+R_{s+1}).
    \]
    These two inequalities show that
    \(C'_{\mathrm{EPCD}}-C_{\mathrm{EPCD}}\) is at least the right-hand side of
    \eqref{eq:upper-qint-lint}. Therefore
    \[
        Q_{\mathrm{int}}-L_{\mathrm{int}}
        \le
        C'_{\mathrm{EPCD}}-C_{\mathrm{EPCD}}.
    \]
    Since \(F_{\mathrm{even}}\le Q_{\mathrm{int}}-L_{\mathrm{int}}\), the lemma
    follows.
\end{proof}

Now the $\frac32$-approximation of our mechanism in the even-agent case follows.

\begin{proof}[Proof of Theorem~\ref{thm:main-upper} in the even case]
    Consider an arbitrary normalized even-agent profile $\mathbf b$ with positive $\mathrm{OPT}$.
    By definition of \(F_{\mathrm{even}}\),
    \(C'_{\mathrm{RD}}+C'_{\mathrm{EPCD}}=3\mathrm{OPT}+F_{\mathrm{even}}\). 
    Hence
    \begin{align*}
        C_{\mathrm{RD}}+C_{\mathrm{EPCD}}-3\mathrm{OPT}
        &=
        F_{\mathrm{even}}
        -
        (C'_{\mathrm{RD}}-C_{\mathrm{RD}})
        -
        (C'_{\mathrm{EPCD}}-C_{\mathrm{EPCD}})\\
       & \le
        F_{\mathrm{even}}-(C'_{\mathrm{EPCD}}-C_{\mathrm{EPCD}}),
    \end{align*}
    where the inequality is due to the facts that the cut saving of any mechanism is non-negative and
    thus $C'_{\mathrm{RD}}-C_{\mathrm{RD}}\ge 0$.
    Since  Lemma~\ref{lem:even-domination} shows that the last expression is at most
    zero, we obtain \(C_{\mathrm{RD}}+C_{\mathrm{EPCD}}\le 3\mathrm{OPT}.\)
    Our mixed mechanism
    \(\frac12 \mathrm{RD}+\frac12 \mathrm{EPCD}\) has circle-metric social cost of
    \(\frac12(C_{\mathrm{RD}}+C_{\mathrm{EPCD}})\), which is at most \(\frac32\mathrm{OPT}\). This establishes the claimed approximation ratio.
\end{proof}

\section{Lower Bound}\label{sec:lower}

We complement our upper bound with a computer-assisted lower bound for general
randomized strategyproof mechanisms on the circle. The construction uses four
agents and improves the previous lower bound of $1.0456$ due to
Meir~\cite{DBLP:conf/sagt/Meir19}.

\begin{theorem}[Computer-assisted lower bound]\label{thm:lower-four}
For $n=4$, every randomized strategyproof mechanism on the circle has
approximation ratio at least
\[
    \rho_0:=\frac{108818734234255}{99999999999758}
    =1.088187342345\ldots.
\]
\end{theorem}

The proof uses a weighted sum of social-cost and incentive constraints on a
finite family of profiles. Symmetry reduces the number of profiles to consider,
and piecewise linearity allows the resulting inequalities to be verified at
finitely many points.

\begin{lemma}[Finite symmetrization]\label{lem:lower-symmetry}
Let $m$ be a positive integer. Using coordinates modulo $1$ on the unit circle,
define
\[
    H_m=\left\{
        h_j^+(x)=x+\frac jm,\quad h_j^-(x)=-x+\frac jm
        : j=0,\ldots,m-1
    \right\}.
\]
Thus $H_m$ consists of $2m$ transformations: rotations by integer multiples of
$\frac1m$ of a full turn, and those rotations combined with the reflection
$x\mapsto -x$. Coordinates differing by an integer denote the same point.
Each transformation preserves distances, and compositions and inverses of
these transformations also belong to $H_m$; this is the finite group of
symmetries used below. A mechanism $f$ is
\emph{$H_m$-equivariant} if, for every $h\in H_m$ and every profile $\mathbf x$,
\[
    f(h\circ\mathbf x)=h\circ f(\mathbf x),
\]
Here $h\circ\mathbf x=(h(x_1),\ldots,h(x_n))$, and the right-hand side is the
distribution obtained by drawing $Y\sim f(\mathbf x)$ and returning $h(Y)$. Thus, transforming all reports by $h$
transforms the output distribution by the same isometry. This is neutrality
restricted to the transformations in $H_m$.
Any randomized strategyproof mechanism can be replaced by an anonymous,
$H_m$-equivariant strategyproof mechanism with no larger approximation ratio.
\end{lemma}
\begin{proof}
Choose an agent permutation and an element of $H_m$ uniformly at random,
apply these transformations to the reported profile, run the original
mechanism, and apply the inverse circle isometry to its output. Each fixed
choice preserves strategyproofness, since it only relabels agents and preserves
distances. It also preserves the approximation guarantee, because both social
cost and its optimum are invariant under these transformations. Averaging
preserves these inequalities and gives the claimed symmetries.
\end{proof}

Two profiles are equivalent for this argument if one can be obtained from the
other by reordering the agents and applying the same transformation in $H_m$
to every report. Anonymity makes the reordering irrelevant, and equivariance
determines how the output distribution changes under the transformation.
We therefore retain one profile from each equivalence class that we use in
the argument. Let $\mathcal P$ denote this finite collection of
\emph{representative profiles}, and let $Y_p\sim f(p)$ be the random output at
$p\in\mathcal P$.

Let $\mathcal E$ index a finite collection of single-agent deviations whose
source and destination profiles are represented in $\mathcal P$. For each
$e\in\mathcal E$, an agent at true location $t_e$ in profile $p_e$ changes its
report. Write $q_e$ for the representative of the resulting profile, and let
$h_e\in H_m$ be the transformation used to map that profile to $q_e$, together
with a reordering of agents. Define $u_e=h_e(t_e)$. Strategyproofness then gives
\begin{equation}\label{eq:lower-incentive}
    \mathbb E[d(t_e,Y_{p_e})]
    \le \mathbb E[d(u_e,Y_{q_e})],
    \qquad e\in\mathcal E.
\end{equation}
The left-hand side is the agent's expected cost when reporting truthfully;
the right-hand side is its expected cost after the deviation. The location
$u_e$ is its true location in the coordinates used for $q_e$, not its new
report. If no rotation or reflection is needed, then $u_e=t_e$.

A \emph{certificate} here is a choice of nonnegative weights for social-cost
and incentive inequalities, together with pointwise lower bounds on their
weighted sum. The following lemma explains how such a certificate yields an
approximation lower bound.

\begin{lemma}[Weighted certificate]\label{lem:lower-certificate}
Assign a nonnegative weight $\alpha_p$ to the social-cost constraint at each
profile $p\in\mathcal P$, and a nonnegative weight $\beta_e$ to each incentive
constraint $e\in\mathcal E$. These weights are coefficients for adding
inequalities; they are not probabilities. Define
\[
    A=\sum_{p\in\mathcal P}\alpha_p\mathrm{OPT}(p)>0,
\]
\[
    D_p(y)=\alpha_p\mathrm{SC}(p,y)
      +\sum_{e:p_e=p}\beta_e d(t_e,y)
      -\sum_{e:q_e=p}\beta_e d(u_e,y).
\]
The quantity $A$ is the weighted sum of optimal social costs. In $D_p$, the
first term comes from the social-cost constraint at $p$. A deviation starting
at $p$ contributes a positive term, while a deviation ending at representative
$p$ contributes a negative term when its right-hand side is moved to the left.

For each profile $p$, let
\[
    H_p=\{h\in H_m: h\circ p\text{ is a reordering of }p\}.
\]
These are the symmetries of that particular profile: after applying $h$, each
location still has the same number of reported agents. For example, a
reflection exchanging two equally populated locations preserves their profile.
Write $|H_p|$ for the number of such transformations, and set
\[
    \overline D_p(y)=\frac1{|H_p|}\sum_{h\in H_p}D_p(h(y)),
    \qquad
    \gamma_p=\min_{y\in G}\overline D_p(y).
\]
Thus $\overline D_p$ averages $D_p$ over the symmetries of $p$, and $\gamma_p$
is the smallest value of this averaged function as the facility moves around
the circle. The value $\gamma_p$ may be negative, since $D_p$ contains signed
terms. Every randomized strategyproof mechanism has approximation ratio at least
\[
    \frac{\sum_{p\in\mathcal P}\gamma_p}{A}.
\]
\end{lemma}
\begin{proof}
By Lemma~\ref{lem:lower-symmetry}, it suffices to consider a symmetrized mechanism
with finite approximation ratio $\rho$. At each profile $p$, its approximation
guarantee gives
\[
    \mathbb E[\mathrm{SC}(p,Y_p)]\le\rho\,\mathrm{OPT}(p).
\]
Multiply these inequalities by $\alpha_p$ and the
inequalities~\eqref{eq:lower-incentive} by $\beta_e$, moving each incentive
constraint's right-hand side to the left. Collecting all terms involving the
same random output $Y_p$ gives
\[
    \sum_{p\in\mathcal P}\mathbb E[D_p(Y_p)]\le \rho A.
\]
For $h\in H_p$, the transformed profile is a reordering of $p$. Anonymity and
equivariance therefore imply that $h(Y_p)$ and $Y_p$ have the same
distribution. Averaging their expectations gives
\[
    \mathbb E[D_p(Y_p)]
    =\mathbb E[\overline D_p(Y_p)]\ge\gamma_p.
\]
Summing over $p$ and dividing by $A$ proves the claim.
\end{proof}

\begin{lemma}[Finite verification on the continuous circle]\label{lem:lower-breakpoints}
For a profile $p$, let $T_p$ contain its reported locations and all true
locations $t_e$ and $u_e$ appearing in its function $D_p$. For a point $t$, its
\emph{antipode} is $t^*=t+\frac12$ modulo $1$, the point halfway around the
circle. Define the finite set
\[
    K_p=\{h(t),h(t^*):t\in T_p,\ h\in H_p\}.
\]
Thus $K_p$ is formed by adding the antipodes of the points in $T_p$ and then
applying all symmetries of $p$. Each minimum in
Lemma~\ref{lem:lower-certificate} satisfies
\[
    \gamma_p=\min_{y\in K_p}\overline D_p(y).
\]
\end{lemma}
\begin{proof}
As $y$ travels around the circle, $d(t,y)$ changes linearly between $t$ and
$t^*$; only at these two points can its slope change. Such a point is called
a \emph{breakpoint}. A weighted sum of distance functions is therefore
linear on every arc containing none of their breakpoints. Averaging over
$H_p$ also transforms these breakpoints. Since $H_p$ contains the inverse of
each of its transformations, all resulting breakpoints belong to $K_p$.
The function $\overline D_p$ is consequently linear between consecutive points
of $K_p$, so its minimum on each such arc is attained at an endpoint.
\end{proof}

\begin{proof}[Proof sketch of Theorem~\ref{thm:lower-four}]
Use the grid
\[
    V=\left\{\frac{j}{120}:j=0,\ldots,119\right\}\subseteq G,
\]
with coordinates interpreted modulo $1$. Here $V^4$ denotes the profiles of
four agents, each reporting one of the 120 grid locations. Select profiles
and single-agent deviations between them, and retain representatives under
agent permutations and $H_{120}$ as described above.

The weights in Lemma~\ref{lem:lower-certificate} are found by linear
programming. For each $p$, the condition
$\overline D_p(y)\ge\gamma_p$ is imposed at every $y\in K_p$. At a fixed
point $y$, this is a linear inequality in the weights and $\gamma_p$, because
all distances are fixed numbers. With the normalization $A=1$, maximizing
$\sum_p\gamma_p$ searches for the strongest certificate supported by the
selected constraints.

The resulting certificate uses 489 representative profiles, with 108 nonzero
social-cost weights and 700 nonzero incentive weights. The numerical weights
are replaced by nonnegative rational weights, and the minima in
Lemma~\ref{lem:lower-breakpoints} are recomputed using exact rational
arithmetic. Dividing every weight by $A$ scales each
$\gamma_p$ by the same factor and leaves the certified ratio unchanged.
The normalized certificate satisfies
\[
    A=1,
    \qquad
    \sum_{p\in\mathcal P}\gamma_p
    =\frac{108818734234255}{99999999999758}.
\]
Exact rational verification checks these identities, the nonnegativity of the
weights, and each selected deviation. The pointwise minima are evaluated at
the breakpoints specified in Lemma~\ref{lem:lower-breakpoints}.
Applying Lemma~\ref{lem:lower-certificate} gives the theorem.
\end{proof}

The same bound holds for $n=4k$, where $k\ge1$ is an integer. To see this,
replace each report of a four-agent profile by $k$ identical reports and apply
an arbitrary $4k$-agent mechanism. If one of the four agents changes its report
from $t$ to $s$, move its $k$ copies one at a time. At each step, ordinary
strategyproofness compares expected distance from the same true location $t$;
adding these inequalities shows that the induced four-agent mechanism is
strategyproof. Social cost and its optimum both scale by $k$, so its
approximation ratio is no larger than that of the original mechanism.

\section{Conclusion}\label{sec:conclusion}

We study randomized strategyproof mechanisms for facility location on a circle under the utilitarian social cost objective. 
We show that a simple {parity-dependent MRP} mechanism built as a mixture of two randomized mechanisms achieves an approximation ratio of \(\frac32\) for every \(n\ge 3\).
For odd agent number \(n\), the mechanism combines the \(\mathrm{RD}\) and \(\mathrm{PCD}\) mechanisms with equal probability; for even \(n\), it combines the \(\mathrm{RD}\) mechanism with \(\mathrm{EPCD}\), a random-deletion extension of \(\mathrm{PCD}\).
The proof refines the circle-cutting approach by explicitly accounting for how much the cut metric overestimates the true circle distance for pairs lying on opposite sides of the cut.
We also show that the \(\frac32\) guarantee is asymptotically tight for MRP as the number of agents grows, and establish a computer-assisted lower bound of \(1.088187\ldots\) on the approximation ratio of any randomized strategyproof mechanism.

Several questions remain open. The most important one is to close the gap between the \(1.088187\ldots\) lower bound and the \(\frac{3}{2}\) upper bound. 
The lower-bound construction in this paper uses four agents, and larger-agent instances may lead to stronger impossibility results. 
Another natural direction is to consider {other classes of} randomized mechanisms.
{Within fixed-weight combinations of \(\mathrm{RD}\) and \(\mathrm{PCD}\), the equal-probability choice is forced by two complementary examples shown in the analysis of Proposition~\ref{prop:tight-examples}, which suggests that it is the best balance within this family.
Thus, improving the upper bound below \(\frac32\) likely requires different types of mechanisms, perhaps one using nonlinear gap-dependent probabilities in the spirit of quadratic circle-distance mechanisms.
}

\bibliographystyle{plain}
\bibliography{myreferences}

\end{document}